\documentclass[twocolumn]{autart}    % Enable this line and disable the 
\usepackage{graphicx}  
\usepackage{cite}
\usepackage{amsmath,amssymb,amsfonts}
\usepackage{graphicx}
\usepackage{subcaption}
\usepackage{textcomp}
\usepackage{mathtools}
\usepackage{xcolor}
\usepackage{fourier}
\usepackage{hyperref}
\usepackage{float}
\usepackage{algorithm}
\usepackage{algorithmic}
\DeclareMathAlphabet{\mathcal}{OMS}{cmsy}{m}{n}
\usepackage{multirow}

\newtheorem{theorem}{\textbf{Theorem}}
\newtheorem{lemma}{\textbf{Lemma}}

\newtheorem{corollary}{\textbf{Corollary}}
\newtheorem{definition}{\textbf{Definition}}
\newtheorem{remark}{\textbf{Remark}}
\newenvironment{proof}{{\emph{\textbf{Proof:}} }}{\hfill $\square$}

\begin{document}

\begin{frontmatter}

\title{Trajectory elongation strategies with minimum curvature discontinuities for a Dubins vehicle}

\thanks[footnoteinfo]{This work was not supported by any organisation.}

\author{Aditya K. Rao}\ead{adikrao@iitk.ac.in},    % Add the 
\author{Twinkle Tripathy}\ead{ttripathy@iitk.ac.in}  % (ead) as shown

\address{Department of Electrical Engineering, Indian Institute of Technology Kanpur, Kanpur-208016, India}  % Please supply                     

\begin{keyword}                           % Five to ten keywords,  
    Dubins Shortest Path, curvature bounded trajectories, elongation strategies, trajectory planning           % chosen from the IFAC 
\end{keyword}                             % keyword list or with the 
                                          % help of the Automatica 
                                          % keyword wizard

\begin{abstract}        
In this paper, we present strategies for designing curvature-bounded trajectories of any desired length between any two given oriented points. The proposed trajectory is constructed by the concatenation of three circular arcs of varying radii. Such a trajectory guarantees a complete coverage of the maximum set of reachable lengths while minimising the number of changeover points in the trajectory to a maximum of two under all scenarios. Additionally, by using the notion of internally tangent circles, we expand the set of Circle-Circle-Circle trajectories to eight kinds, consisting of $\{LLL,LLR,LRR,LRL,RRL,RLL,RLR,RRR\}$ paths. The paper presents a mathematical formulation of the proposed trajectory and the conditions for the existence and classification of each kind of trajectory. We also analyse the variation of the length of the trajectory using suitable elongation strategies and derive the set of reachable lengths for all pairs of oriented points. Finally, the results of this paper are illustrated using numerical simulations.
\end{abstract}

\end{frontmatter}

\section{Introduction}

The problem of planning trajectories between two given points for an autonomous vehicle moving at a constant speed has been explored extensively in literature of guidance and control. It finds its applications in a variety of fields such as parking problems\cite{b29}, warehouse automation\cite{b30}, missile guidance\cite{b27}, \textit{etc.} The feasibility of such problems towards real world applications leads to additional requirements such as curvature-boundedness, desired lengths of the trajectories, directions of motion at initial and final points\cite{b12}, and optimal energy consumption \cite{b28}.

% must be added while designing trajectories.
%A well-researched problem in the control and guidance arena is planning of trajectories for an autonomous vehicle moving at a constant speed between two given points. 

The shortest trajectory between any two oriented points in $\mathbb{R}^2$ is either of the form Circle-Circle-Circle $(CCC)$ or Circle-Straight Line-Circle $(CSC)$ (\cite{dubins},\cite{b14}). This trajectory is called Dubins Shortest Path. A relaxed problem is to find trajectories of desired lengths between an oriented point and a fixed point in $\mathbb{R}^2$. The shortest path in such a scenario is of the form Circle-Circle $(CC)$ or Circle-Straight Line $(CS)$ as shown in \cite{bui1994accessibility}. Elongation of certain subsections of such a minimum length trajectory is done to get trajectories of desired lengths. Without fixing the tangent vector at the final point, the authors in \cite{b12} and \cite{b13} present multiple such elongation strategies. However, it is shown in \cite{b19} that between some initial oriented points and final points, there exist no curvature-bounded paths for a certain range of lengths. This classification of the set of reachable lengths is elaborated in \cite{b19} and elongation strategies are presented for all the cases. %on it to find a trajectory of any particular length.
% It is shown in \cite{dubins}, \cite{b14} that theminimum length trajectory between any two points in $\mathbb{R}^2$ with tangent vectors fixed on both of them, i.e. oriented points, is composed of a series of circles and straight lines and is called Dubins shortest path.
% Moreover, the results discussed so far do not impose a constraint on the final velocity vector.

 The authors in \cite{b15} discuss elongation strategies achieved by increasing the radii of curvature of the terminal circles when the Dubins Shortest Path is of the form $CSC$. Alternatively, the authors in \cite{b22} propose replacing the straight-line segment with an elongated path. The use of clothoid arcs of arbitrary lengths is proposed in \cite{b1} for trajectory elongation while also ensuring a continuity in the curvature profile. The authors in \cite{chen_elongation} provide a comprehensive analysis of the set of reachable lengths given any two oriented points and propose elongation strategies for all pairs of oriented points. They also highlight the cases when certain lengths are not reachable for some pairs of oriented points. In \cite{b25}, the Dubins Shortest Path is extended to three-dimensional space, followed by the introduction of trajectory elongation methods to attain any desired length. The authors in \cite{patsko2022three} analyse the set of reachable oriented points from any initial oriented point at time $t_f$ for a Dubins vehicle with symmetric bounds on the control input. Further, in the case of asymmetric bounds on the control input, the set of reachable oriented points is derived in \cite{patsko2023threeasym}. Note that, in general, the use of elongation based strategies increases the number of changeover points. The strategies presented in \cite{b15,b22,b1,chen_elongation} increase the number of changeover points by as large as six. To address this problem, in one of our earlier works \cite{RAO}, we propose to construct trajectories of desired lengths using exactly two circles of varying radius minimising the changeover points to exactly one always.
  
Geometrical curves other than circles and straight lines are also used in the literature to construct trajectories of desired lengths. The use of elliptical curves and Bezier curves is proposed in \cite{b8} and \cite{b26}, respectively, to achieve trajectories of desired lengths. However, the use of geometrical curves, other than the circles and straight lines, limits the set of reachable lengths. Alternate methods have also been explored to achieve a trajectory of desired length between any two oriented points. The authors in \cite{xu1999curve} construct the trajectory as a polynomial expression and find its parameters that satisfy various constraints. Optimal control theory is employed in \cite{b23} and \cite{b24} to develop closed-form Impact Time Control Guidance (ITCG) laws, allowing a vehicle to manoeuvre between two oriented points with a fixed time of flight. In contrast, \cite{b2} introduces a structure-homotopy-based planner that generates trajectories by focusing on endpoint conditions rather than relying on elongation strategies.
%
%This structured-homotopy based planner is further extended to develop a consensus based trajectory planner for a multiple vehicle system in \cite{b3}.

In this work, we focus on planning trajectories of desired lengths between any two given oriented points while minimising the number of changeover points and ensuring the maximum coverage of the set of reachable lengths. The major contributions of this work are as follows:
\begin{itemize}
    \item \textit{Design of paths of desired lengths:} We propose to construct a feasible trajectory by concatenating three circles. We show that there always exist infinitely many such curvature-bounded trajectories of varying lengths between any two oriented points and present analytical results for the same.
    \item \textit{Internally tangent trajectories:} The Circle-Circle-Circle Dubins paths are of form $\{RLR,LRL\}$ in literature. The proposed solution expands the Dubins path of form Circle-Circle-Circle $(CCC)$ to eight forms, namely $\{LLL,LLR,LRR,LRL,RRL,RLL,RLR,RRR\}$, by considering circles that are internally tangent. We also show that the Circle-Circle trajectories of form $\{LL,LR,RR,RL\}$ presented in \cite{RAO} emerge naturally as a subset of the proposed solution. 
    \item \textit{Set of reachable lengths:} We find the set of reachable lengths for these trajectories and show that it is equal to the maximum set of reachable lengths. Given a reachable length, we show that there exist multiple curvature-bounded trajectories of a desired length. 
    \item \textit{Reduction in the number of points of curvature discontinuity:} 
    The proposed trajectory always has a maximum of two changeover points for any two arbitrarily orientated points and a given desired length. This minimisation is achieved without any reduction in the maximum set of reachable lengths.
\end{itemize}

The paper is organised as follows: Section \ref{section: prelim} presents some preliminary technical results that are used subsequently in the paper. Section \ref{sec: problem} presents the problem statement and describes the proposed method of trajectory design with the motivation behind it. Section \ref{section: existence} provides a mathematical description of the proposed Circle-Circle-Circle trajectory and the conditions necessary for its existence. Building on these results, Section \ref{section: traj of desired length} explores elongation strategies and their impact on the attributes of the proposed trajectories. Section \ref{sec: reachability set} discusses the set of reachable lengths for various pairs of oriented points. Section \ref{sec: simulation} presents numerical simulations to illustrate the theoretical findings. Finally, Section \ref{sec: conclusion} concludes the paper and suggests potential directions for future research.

\section{Preliminaries and Notations}
\label{section: prelim}
Trajectories constructed using circles have been utilised frequently in path planning problems. They provide simplicity in design and require constant lateral acceleration in physical implementation. We define a curvature-bounded trajectory as a function $\Lambda$ and the minimum turn radius allowed for any such trajectory as $r_{\min}$.
\begin{definition}
\label{def:lambda}
    Given two oriented points ${A}=(\mathbf{a},\alpha)$ and ${B}=(\mathbf{b},\beta)$, $\Lambda:[0,s]\longrightarrow\mathbb{R}^2$ denotes a feasible curvature-bounded trajectory connecting the oriented points such that
    \begin{itemize}
        \item $\Lambda(t)$ is parameterized by arc length and is $C^1$ and piece-wise $C^2$.
        \item $|{\Lambda^{'}(t)}| = 1 \forall
t\in[0, s]$
        \item $\Lambda(0)=\mathbf{a},\Lambda(s)=\mathbf{b},\Lambda'(0)=(\cos\alpha,\sin\alpha),\Lambda'(s)=(\cos\beta,\sin\beta)$    \item $||\Lambda^{''}(t)||\leq1/r_{\min}$, $t\in[0,s]$ when defined, 
        \end{itemize}
    where $r_{\min}>0$. The length of the trajectory is denoted by $l(\Lambda)$.
\end{definition}
 We define $\mathcal{C}^r_P$ and $\mathcal{C}^l_P$ as two tangential circles of radius $r_{\min}$ at an oriented point $P$ with their centres at $\mathbf{c}^r_P$ and $\mathbf{c}^l_P$, respectively, as shown in Fig. \ref{fig: left right circles}. These circles correspond to a right turn (denoted by $R$) and a left turn (denoted by $L$) on a circular arc of radius $r_{\min}$. Further, we define the function $d:\mathbb{R}^2\times\mathbb{R}^2\longrightarrow\mathbb{R}$ that gives the Euclidean distance between any two points in $\mathbb{R}^2$ plane. 

\begin{figure}[h]
    \centering
    \includegraphics[scale=0.5]{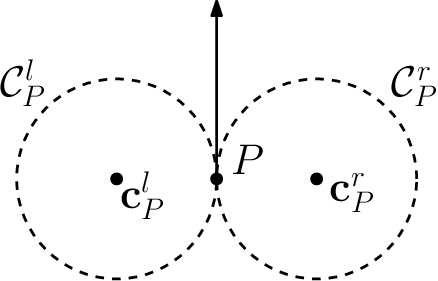}
    \caption{Two circles $\mathcal{C}^r_P$ and $\mathcal{C}^l_P$ at oriented points $P$}
    \label{fig: left right circles}
\end{figure}

The radius of a circle is generally a positive real number. However, we denote the radius of a circle as $r\in\mathbb{R}$ within the context of this paper with the following attributes.
\begin{definition}
\label{lem:neg_radius}
For any value of the radius $r\in\mathbb{R}$, the magnitude of curvature is given by $1/|r|$, and the orientation of motion on the trajectory $\Lambda$ is counter-clockwise for $r>0$ and is clockwise for $r<0$.
\end{definition}
We now present some useful results from \cite{RAO} on the tangency of circles and the motion over trajectories formed by such tangential circles. The statements of these results have been appropriately re-phrased for the context of this paper. 
\begin{lemma} 
\label{lem:tangency centre disctance}
Any two tangential circles must satisfy the condition,
\begin{equation}
    d(\mathbf{o_a},\mathbf{o_b})=|r_a-r_b|
    \label{eq:centre disctance}
\end{equation}
where $\mathbf{o_a}$ and $\mathbf{o_b}$ are the coordinates of centres of the circles and $r_a$ and $r_b$ are the radii, respectively.
\end{lemma}
The relation in \eqref{eq:centre disctance} for the distance between the centres of any two (internally or externally) tangential circles is within the context of the radius given in Definition \ref{lem:neg_radius}. %The relation holds for both internal and external tangency of the circles.
\begin{lemma}
\label{lem:tangency relation circles}
If two circles $\mathcal{C}_a$ and $\mathcal{C}_b$ are externally tangent, their radii, $r_a$ and $r_b$, respectively, have opposite signs, otherwise they have the same signs. 
\end{lemma}
\begin{lemma}
    If any two circular arcs are externally tangent, then, the orientation of the motion switches from clockwise to anti-clockwise or vice-versa along the trajectory at the point of tangency. On the contrary, if the circles are internally tangent, the orientation of the motion remains the same. (Fig. \ref{fig:typesoftangency})
    \label{lem: orientation change}
\end{lemma}
\begin{figure}[h]
\centering
    \begin{subfigure}{0.2\textwidth}
    \centering
        \includegraphics[width=0.75\textwidth]{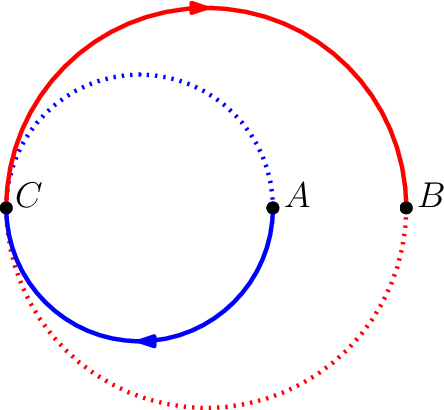} 
        \caption{Internally Tangent Arcs}
        \label{fig:external orient}
    \end{subfigure}
    \hskip2em 
     \begin{subfigure}{0.2\textwidth}
     \centering
        \includegraphics[width=\textwidth]{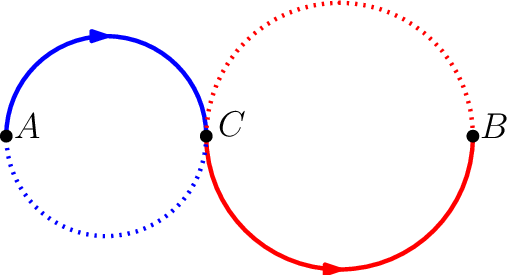} 
        \caption{Externally Tangent Arcs}
        \label{fig:internal orient}
    \end{subfigure}
    \caption{Change of orientation in Internally and Externally Tangent Arcs}
    \label{fig:typesoftangency}
\end{figure}

We know from \cite{dubins} that the shortest curvature-bounded trajectory between any two oriented points $A$ and $B$ is either of the form Circle-Straight Line-Circle ($CSC$) or Circle-Circle-Circle ($CCC$), where C denotes a circular arc of radius $r_{\min}$ and S denotes a straight-line segment. We denote this shortest trajectory by $\Lambda_m$ and its length by $l_m$. There exist two types of $CCC$ paths, namely $\{LRL,RLR\}$; and four types of $CSC$ paths, namely $\{LSR,LSL,RSL,RSR\}$.

\section{Problem Statement}
\label{sec: problem}
Consider two points $\mathbf{a}$ and $\mathbf{b}$ in $\mathbb{R}^2$ space. Without loss of generality, we assume that the final point $\mathbf{b}$ lies at the origin. Let $\alpha$ and $\beta$ be the angles that the velocity vectors make at $\mathbf{a}$ and $\mathbf{b}$ with respect to the positive X-axis, respectively, such that $\alpha,\beta \in [0,2\pi)$. Using them, we define the tuples ${A}(\mathbf{a},\alpha)$ and ${B}(\mathbf{b},\beta)$ to be two oriented points. \textit{The objective of the paper is to construct a trajectory of any arbitrary length $l_o$ between any two given oriented points $A$ and $B$.}  Further, additional constraints of curvature-boundedness and minimum curvature discontinuities are imposed on the trajectory to facilitate practical implementation.
\begin{figure}[h]
    \centering
    \includegraphics[scale=0.5]{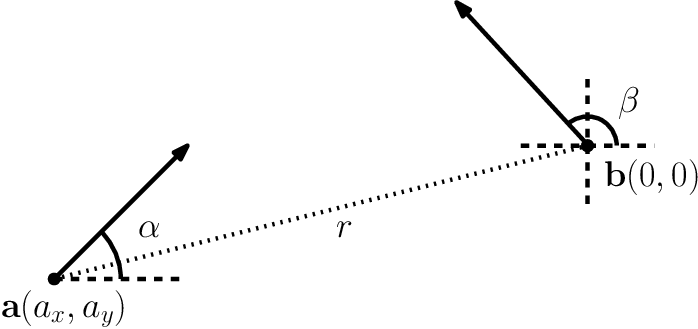}
    \caption{Oriented points $A(\mathbf{a},\alpha)$ and $B(\mathbf{b},\beta)$ for trajectory design}
    \label{fig:oriented points}
\end{figure}

Any feasible trajectory between $A$ and $B$ has five attributes to satisfy: the initial heading angle ($\alpha$), the final heading angle ($\beta$), the length of the trajectory ($l_o$) and the relative location of points $\mathbf{a}$ and $\mathbf{b}$ in $\mathbb{R}^2$ space. \textit{To impose all of the above conditions, we design a trajectory formed by concatenating three circular arcs with exactly two pairs of circular arcs tangent to each other.} We refer to such trajectories as Circle-Circle-Circle trajectories. The circles are denoted by $\mathcal{C}_1$, $\mathcal{C}_2$ and $\mathcal{C}_3$ and their respective radii by $r_1$, $r_2$ and $r_3$. In any proposed feasible trajectory,%such that
    \begin{enumerate}
    \item[i.] the trajectory begins at the point $\mathbf{a}$ on $\mathcal{C}_1$,
    \item[ii.] the circles $\mathcal{C}_1$ and $\mathcal{C}_2$ are tangent at a point $\mathbf{c_1}$ and the circles $\mathcal{C}_2$ and $\mathcal{C}_3$ are tangent at a point $\mathbf{c_2}$. These points are called the \textit{changeover points}.
    \item[iii.] the trajectory ends at point the $\mathbf{b}$ on $\mathcal{C}_3$, and,
    \item[iv.] the overall trajectory is composed of the union of three circular arcs given by $\widearc{\mathbf{ac_1}} \cup \widearc{\mathbf{c_1c_2}}\cup\widearc{\mathbf{c_2b}}$.
\end{enumerate}
\begin{figure}[h]
    \centering
    \includegraphics[scale=0.45]{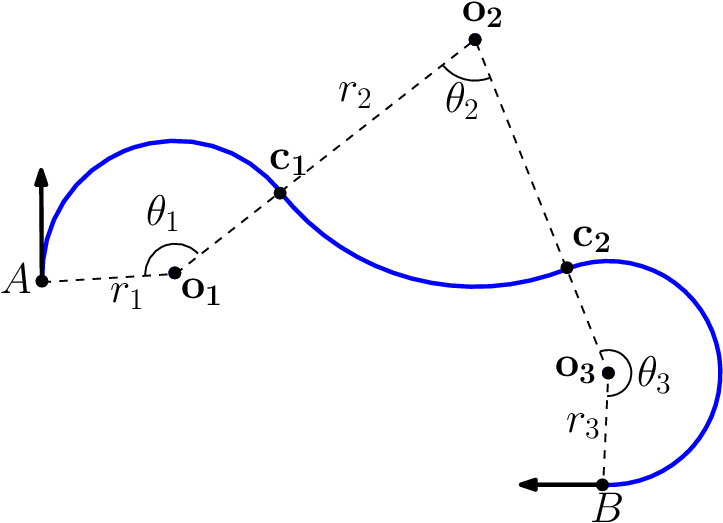}
    \caption{Parameters of Circle-Circle-Circle Trajectory}
    \label{fig:traj_cc}
\end{figure}
Fig. \ref{fig:traj_cc} illustrates a Circle-Circle-Circle trajectory. We denote such trajectories by $\mathbf{C}^{r_1}\mathbf{C}^{r_2}\mathbf{C}^{r_3}$ where $r_1,r_2\text{ and }r_3$ correspond to the radius of the three circles, respectively. This is a generalisation of the commonly used notation $CCC$ which represents a Dubins Paths comprising of three circles of curvature $1/r_{\min}$. We now highlight the motivation for using such trajectories:
\begin{itemize}
    \item In one of our previous works \cite{RAO}, Circle-Circle trajectories are proposed as feasible trajectories in the given framework. The degrees of freedom of such a trajectory are exactly equal to the attributes of the desired trajectory. However, the set of reachable lengths is a subset of the maximum set of reachable lengths given by Theorem \ref{thm: maximal reachable lengths}. We show eventually that the proposed Circle-Circle-Circle trajectories overcome this limitation. 
    \item Three circles of varying radii lead to nine degrees of freedom in the overall trajectory. With two tangency constraints, we get seven degrees of freedom in a Circle-Circle-Circle trajectory. The extra degrees of freedom are advantageous as they lead to the existence of multiple trajectories between $A$ and $B$.
\end{itemize}

The following section deals with the mathematical formulation to design $\mathbf{C}^{r_1}\mathbf{C}^{r_2}\mathbf{C}^{r_3}$ trajectories between any two oriented points.
%%%%%%%%%%%%%%%%%%%%%%%%%%%%%%%%%%%%%%%%%%%%%%%%%%%%%
\section{Existence of Circle-Circle-Circle trajectory}
\label{section: existence}
 
In this section, we explore the existence of a general $\mathbf{C}^{r_1}\mathbf{C}^{r_2}\mathbf{C}^{r_3}$ trajectory while relaxing the constraint of a desired length. Then, given the oriented points $A(\mathbf{a},\alpha)$ and $B(\mathbf{b},\beta)$, consider the circles $\mathcal{C}_1$ and $\mathcal{C}_3$. Their centres must lie on the lines normal to their respective heading vectors at the terminal points $\mathbf{a}=(a_x,a_y)$ and $\mathbf{b}=(0,0)$.  Then,
\begin{subequations}
\begin{align}
    \mathbf{o_1}&=(a_x-r_1\sin\alpha,a_y+r_1\cos\alpha)\\
    \mathbf{o_3}&=(-r_3\sin\beta,r_3\cos\beta)
\end{align}
\label{eq:centres o1 and o3}
\end{subequations}
As mentioned in Definition \ref{lem:neg_radius}, we interpret $r_1$ and $r_3$ (and later $r_2$) as not just the radii of the respective circles but also the orientations of motion on them. Eqn. \eqref{eq:centres o1 and o3} can be viewed as the locus of the centres $\mathbf{o_1}$ and $\mathbf{o_3}$ parameterised by $r_1\in\mathbb{R}$ and $r_3\in\mathbb{R}$. Note that eqn. \eqref{eq:centres o1 and o3} holds irrespective of the choice of the intermediate curve. This implies that any curve which satisfies the the five required degrees of freedom can be used here. We choose to use a circular path so that the Dubins Shortest Path is achievable by the resulting trajectory. Depending upon the curve used, different existence conditions emerge. It is not always guaranteed that such a path will exist. Next, we proceed to find the conditions for the existence of the proposed $\mathbf{C}^{r_1}\mathbf{C}^{r_2}\mathbf{C}^{r_3}$ trajectories.

\subsection{Locus of the centre of circle \texorpdfstring{$\mathcal{C}_2$}{Lg}}
\label{subs: locus of centre}
In order to determine the conditions to check the existence of $\mathbf{C}^{r_1}\mathbf{C}^{r_2}\mathbf{C}^{r_3}$ trajectories, we focus on $\mathcal{C}_2$ while keeping $\mathcal{C}_1$ and $\mathcal{C}_3$ fixed. In other words, we assume that $r_1$ and $r_3$ take some finite values in $\mathbb{R}$.
\begin{theorem}
     For any given circles $\mathcal{C}_1$ and $\mathcal{C}_3$, the locus of the centre of $\mathcal{C}_2$, denoted by $\mathbf{o_2}$, is a hyperbola $\mathcal{H}$ defined by,
    \begin{equation}
        \mathcal{H}=\{\mathbf{o_2}\in\mathbb{R}^2 \mid ~~|d(\mathbf{o_2},\mathbf{o_1})-d(\mathbf{o_2},\mathbf{o_3})|=|r_1-r_3|\}
        \label{eq: hyperbolic relation of centres}
    \end{equation}
where $\mathbf{o_1}$ and $\mathbf{o_3}$ are centres of $\mathcal{C}_1$ and $\mathcal{C}_3$, respectively, in the feasible $\mathbf{C}^{r_1}\mathbf{C}^{r_2}\mathbf{C}^{r_3}$ trajectory and $r_1$ and $r_3$ are the corresponding radii. 
\label{thm: hyperbola of centre}
\end{theorem}
\begin{proof}
    Let us denote the radius of $\mathcal{C}_2$ as $r_2\in\mathbb{R}$. We proceed with the proof by considering two cases:
    \begin{itemize}
        \item [i.] $r_1r_3>0$: Consider $r_1>0$ and $r_3>0$. Thus, there is a left turn on both $\mathcal{C}_1$ and $\mathcal{C}_3$. For this to happen, $\mathcal{C}_2$ must be chosen such that the orientation of motion changes at both $\mathbf{c_1}$ and $\mathbf{c_2}$ or at neither of them. Thus, $\mathcal{C}_2$ must be either externally or internally tangent to both $\mathcal{C}_1$ and $\mathcal{C}_3$ from Lemma \ref{lem: orientation change}. For internal tangency, $r_2>0$ from Lemma \ref{lem:tangency relation circles}. Consequently, $d(\mathbf{o_2},\mathbf{o_1})=r_2-r_1$ and $d(\mathbf{o_2},\mathbf{o_3})=r_2-r_3$ which implies $ d(\mathbf{o_2},\mathbf{o_1})-d(\mathbf{o_2},\mathbf{o_3})=r_3-r_1$. For external tangency, $r_2<0$ from Lemma \ref{lem:tangency relation circles}. Hence, $d(\mathbf{o_2},\mathbf{o_1})=r_1-r_2$ and $d(\mathbf{o_2},\mathbf{o_3})=r_3-r_2$ which implies $ d(\mathbf{o_2},\mathbf{o_1})-d(\mathbf{o_2},\mathbf{o_3})=r_1-r_3$. Combining them, we get $\left|d(\mathbf{o_2},\mathbf{o_1})-d(\mathbf{o_2},\mathbf{o_3})\right|=|r_1-r_3|$. We get a similar relation for the case $r_1<0$ and $r_3<0$.

        \item [ii.] $r_1r_3<0$: Consider $r_1>0$ and $r_3<0$. Thus, there is a left turn on $\mathcal{C}_1$ and a right turn on $\mathcal{C}_3$. For this to happen, $\mathcal{C}_2$ must be chosen such that the orientation changes at exactly one point amongst $\mathbf{c_1}$ and $\mathbf{c_2}$. From Lemma \ref{lem: orientation change}, $\mathcal{C}_2$ must be externally tangent to one circle and internally tangent to the other. Thus, for internal tangency at $\mathcal{C}_1$ and external tangency at $\mathcal{C}_3$, $r_2>0$ from Lemma \ref{lem:tangency relation circles}. Hence, $d(\mathbf{o_2},\mathbf{o_1})=r_2-r_1$ and $d(\mathbf{o_2},\mathbf{o_3})=r_2-r_3$ which implies $d(\mathbf{o_2},\mathbf{o_1})-d(\mathbf{o_2},\mathbf{o_3})=r_3-r_1$. For the other case where internal tangency occurs with $\mathcal{C}_3$ and external tangency with $\mathcal{C}_1$, $r_2<0$ from Lemma \ref{lem:tangency relation circles}. Consequently, $d(\mathbf{o_2},\mathbf{o_1})=r_1-r_2$ and $d(\mathbf{o_2},\mathbf{o_3})=r_3-r_2$ which implies $d(\mathbf{o_2},\mathbf{o_1})-d(\mathbf{o_2},\mathbf{o_3})=r_1-r_3$. Combining them, we get $\left|d(\mathbf{o_2},\mathbf{o_1})-d(\mathbf{o_2},\mathbf{o_3})\right|=|r_1-r_3|$. We get similar relation for the case $r_1<0$ and $r_3>0$.
    \end{itemize}
The relation $\left|d(\mathbf{o_2},\mathbf{o_1})-d(\mathbf{o_2},\mathbf{o_3})\right|=|r_1-r_3|$ represents the hyperbola $\mathcal{H}$ defined in eqn. \ref{eq: hyperbolic relation of centres} whose foci are at $\mathbf{o_1}$ and $\mathbf{o_3}$. It is depicted in Fig. \ref{fig: hyperbola locus}. Hence, proved.
\end{proof}

The hyperbola $\mathcal{H}$ is important for the analysis of the existence of $\mathbf{C}^{r_1}\mathbf{C}^{r_2}\mathbf{C}^{r_3}$ trajectories. For every point on $\mathcal{H}$, a unique $\mathbf{C}^{r_1}\mathbf{C}^{r_2}\mathbf{C}^{r_3}$ trajectory exists between given $A$ and $B$. We highlight some more attributes of $\mathcal{H}$ in the following subsection.%Fig. \ref{fig: hyperbola locus} shows $\mathcal{H}$ for two given circles. 
\begin{figure}[h]
    \centering
    \includegraphics[scale=0.4]{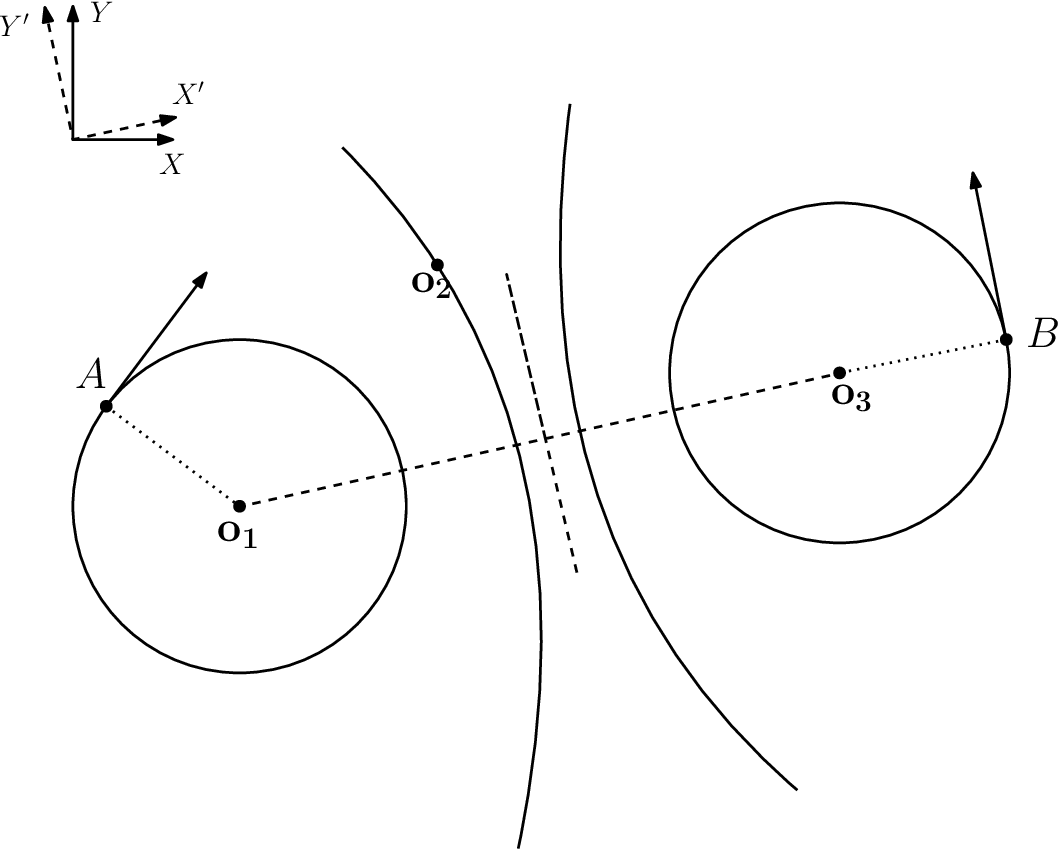}
    \caption{Locus of $\mathbf{o_2}$ }
    \label{fig: hyperbola locus}
\end{figure}

\subsection{Properties of the hyperbola \texorpdfstring{$\mathcal{H}$}{Lg}}
\label{subs: math of H}
The analytical expression for $\mathcal{H}$ is significant from theoretical as well as computational perspectives. We start the analysis with a discussion of some useful properties of $\mathcal{H}$.
\begin{itemize}
    \item The focal length and the length of the semi-major and the semi-minor axes are given by  $\Bar{c}=\frac{d(\mathbf{o_3},\mathbf{o_1})}{2}$, $\Bar{a}=\frac{|r_3-r_1|}{2}$ and $\Bar{b}=\sqrt{\Bar{c}^2-\Bar{a}^2}$, respectively. 
    \item We denote the coordinate frame of reference along the semi-major and semi-minor axes as $X'-Y'$. Thereafter, we define $\hat{n}=[n_x,n_y]^T=\frac{\mathbf{o_3-o_1}}{d(\mathbf{o_3},\mathbf{o_1})}$ as a unit vector along the major axis. Then, $R=\begin{bmatrix}
    n_x & -n_y\\
    n_y & n_x
\end{bmatrix}$ is the rotation matrix between the frames  $X-Y$ and  $X'-Y'$ (see Fig. \ref{fig: hyperbola locus}). 
\item The parametric expression of $\mathcal{H}$ in $X'-Y'$ axis is $\mathbf{o_2}=[a\sec k,b\tan k]^T$ where $k\in[-\pi/2,3\pi/2)$. Through a sequence of rotation and translation the parametric coordinates of $\mathbf{o_2}$ in $X-Y$ coordinates are given by 
\begin{equation}
    \mathbf{o_2}=R\begin{bmatrix}
        a\sec k\\
        b\tan k
    \end{bmatrix} + \frac{\mathbf{o_1+o_3}}{2}
    \label{eq: o2 centre of c2}
\end{equation}
\item Thus, all the points on $\mathcal{H}$  can be parameterised by $k\in[-\pi/2,3\pi/2)$. The range of $k$ is chosen such that a continuous parameterization is achieved for each branch. The right branch is parameterised by $k\in[-\pi/2,\pi/2)$ and the left branch by $k\in[\pi/2,3\pi/2)$. 
\end{itemize}
The existence of the hyperbola is inherently related to $\{A,B\}$ and the values of $\{r_1,r_3\}$. The same is highlighted in the next result.
\begin{theorem}
    Given two oriented points $A$ and $B$, the hyperbola $\mathcal{H}$ exists for any $r_1$ and $r_3$ if 
    \begin{equation}
        d(\mathbf{o_3},\mathbf{o_1})>|r_3-r_1|.
        \label{eq: existence of Hyperbola}
    \end{equation}
    where $\mathbf{o_1}$ and $\mathbf{o_3}$ are given by \eqref{eq:centres o1 and o3}.
    \label{corr: existence of H}
\end{theorem}
\begin{proof}
    From the definition of $\mathcal{H}$, we know that the length of the semi-minor axis $\Bar{b}$ is a positive real number. It then follows that $\Bar{c}>\Bar{a}$. This implies that $d(\mathbf{o_3},\mathbf{o_1})>|r_3-r_1|$. Hence, proved.
\end{proof}

Based on the proof of Theorem \ref{corr: existence of H}, two interesting cases arise for $\mathcal{H}$: 
\begin{itemize}
    \item[i.] In the limiting case of $\Bar{b}=0$, $d(\mathbf{o_3},\mathbf{o_1})=|r_3-r_1|$. This is the condition for the existence of a Circle-Circle trajectory and results in the hyperbolic relation presented in \cite{RAO}. Equivalently, by using appropriate values of $r_1$ and $r_3$, a Circle-Circle trajectory can be designed within the given framework.
    \item[ii.] If $r_1=r_3$, then eqn. \eqref{eq: hyperbolic relation of centres} reduces to two super-imposed straight lines. This line is the perpendicular bisector of the line segment joining $\mathbf{o_1}$ and $\mathbf{o_3}$. 
\end{itemize}
To summarise, given two oriented points $A$ and $B$, consider the ordered pair $(r_1,r_3)\in\mathbb{R}^2$. The equation $d(\mathbf{o_3},\mathbf{o_1})=|r_3-r_1|$ gives a hyperbolic relation in $r_1$ and $r_3$ which partitions the $\mathbb{R}^2$ space into two regions as shown in Fig. \ref{fig: partition r1 r3}. Note that $(r_1,r_3)=(0,0)$ always satisfies \eqref{eq: existence of Hyperbola}. Thus, the region highlighted in green in Fig. \ref{fig: partition r1 r3} denotes the allowed values of $r_1$ and $r_3$ to construct a feasible $\mathbf{C}^{r_1}\mathbf{C}^{r_2}\mathbf{C}^{r_3}$ trajectory. The boundary results in a Circle-Circle trajectory. 
\begin{figure}[h]
    \centering
    \includegraphics[scale=0.65]{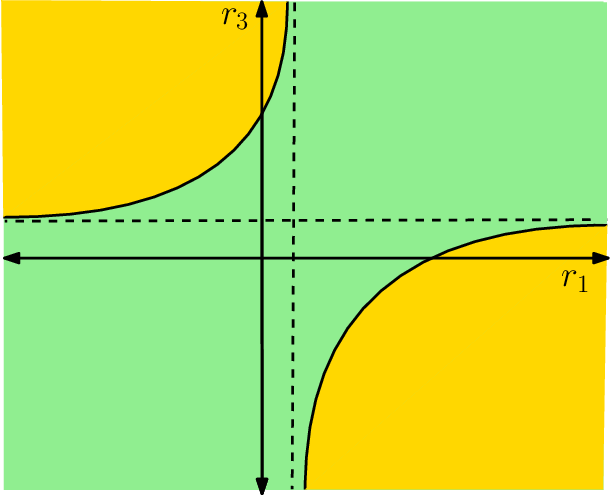}
    \caption{Allowed values of $(r_1,r_3)$ for the existence of $\mathbf{C}^{r_1}\mathbf{C}^{r_2}\mathbf{C}^{r_3}$ trajectory (in green)}
    \label{fig: partition r1 r3}
\end{figure}

The design of $\mathbf{C}^{r_1}\mathbf{C}^{r_2}\mathbf{C}^{r_3}$ trajectories can be effectively visualized through the variation of the three parameters $\{r_1,r_3,k\}$. Theorem \ref{corr: existence of H} gives the sufficient condition for the existence of such a trajectory. Since each circle can either be a left or right turn, it follows naturally that eight permutations are possible for a $\mathbf{C}^{r_1}\mathbf{C}^{r_2}\mathbf{C}^{r_3}$ trajectory. Next, we proceed to classify these trajectories based on the parameters $\{r_1,r_3,k\}$.
\subsection{Classification of trajectories and radius of circle \texorpdfstring{$\mathcal{C}_2$}{Lg}}
\label{subs: classification of CCC}
 Any $\mathbf{C}^{r_1}\mathbf{C}^{r_2}\mathbf{C}^{r_3}$ trajectory can be classified into eight types: $\{LLL,RRR,LLR,RRL,LRR,RLL,LRL,RLR\}$ based upon the orientation of motion on each circular arc. The kind of tangency between any two consecutive circles, which is in turn dependent on the parameters $\{r_1,r_3,k\}$, determines the type of  overall trajectory. The following results illustrate the nature and significance of the tangency between the circles in a given  $\mathbf{C}^{r_1}\mathbf{C}^{r_2}\mathbf{C}^{r_3}$ trajectory. 
\begin{lemma}
    For any given $r_1$ and $r_3$ such that $r_1r_3<0$, 
    \begin{enumerate}
        \item[i.] if $k\in(-\pi/2,\pi/2),~\mathcal{C}_2$ is externally tangent to $\mathcal{C}_1$ and internally tangent to $\mathcal{C}_3$, and
        \item[ii.] if $k\in(\pi/2,3\pi/2),~\mathcal{C}_2$ is internally tangent to $\mathcal{C}_1$ and externally tangent to $\mathcal{C}_3$.
    \end{enumerate}
    \label{lem: classification 1}
\end{lemma}
\begin{proof}
    Consider $r_1>0$ and $r_3<0$. We have shown in proof of Theorem \ref{thm: hyperbola of centre} that for the case of internal tangency at $\mathcal{C}_1$ and external tangency at $\mathcal{C}_3$, $d(\mathbf{o_2},\mathbf{o_1})-d(\mathbf{o_2},\mathbf{o_3})=r_3-r_1<0$. Clearly, $\mathbf{o_2}$ belongs to the left branch as it is farther from $\mathbf{o_3}$ than $\mathbf{o_1}$. For the alternate case,  $d(\mathbf{o_2},\mathbf{o_1})-d(\mathbf{o_2},\mathbf{o_3})=r_1-r_3>0$. This corresponds to $\mathbf{o_2}$ belonging to the right branch as it is farther from $\mathbf{o_1}$ than $\mathbf{o_3}$. Thus, each branch of the hyperbola results in one particular kind of tangency for any $\mathbf{o_2}$ on that branch. The range of $k$  follows automatically.
We can similarly prove the case with $r_1<0$ and $r_3>0$.
Hence, proved.       
\end{proof}

\begin{lemma}
    For any given $r_1$ and $r_3$ such that $r_1r_3>0$ and $|r_1|>|r_3|$, \begin{enumerate}
        \item[i.] if $k\in(-\pi/2,\pi/2),~\mathcal{C}_2$ is externally tangent to both $\mathcal{C}_1$ and $\mathcal{C}_3$, else 
        \item[ii.] if $k\in(\pi/2,3\pi/2),~\mathcal{C}_2$ is internally tangent to both $\mathcal{C}_3$ and $\mathcal{C}_1$.
    \end{enumerate}
    For $|r_1|<|r_3|$, the branches switch.
    \label{lem: classification 2}
\end{lemma}
\begin{proof}
    Consider the case $r_1>0$ and $r_3>0$. We have shown in proof of Theorem \ref{thm: hyperbola of centre} that for internal tangency to both $\mathcal{C}_1$ and $\mathcal{C}_3$, $d(\mathbf{o_2},\mathbf{o_1})-d(\mathbf{o_2},\mathbf{o_3})=r_3-r_1<0$ since $|r_1|>|r_3|$. Clearly, $\mathbf{o_2}$ belongs to the left branch as it is farther from $\mathbf{o_3}$ than $\mathbf{o_1}$. For the case of external tangency, $d(\mathbf{o_2},\mathbf{o_1})-d(\mathbf{o_2},\mathbf{o_3})=r_1-r_3>0$ since $|r_1|>|r_3|$. This corresponds to $\mathbf{o_2}$ belongs to the right branch as it is farther from $\mathbf{o_1}$ than $\mathbf{o_3}$. If $|r_1|<|r_3|$, the branches switch. The range of $k$ follows automatically. We can similarly prove for the case with $r_1<0$ and $r_3<0$
\end{proof}

Using Lemma \ref{lem: classification 1} and Lemma \ref{lem: classification 2}, the classification of $\mathbf{C}^{r_1}\mathbf{C}^{r_2}\mathbf{C}^{r_3}$ trajectories for various values of $\{r_1,r_3,k\}$ can done as shown in Table \ref{table: tangency value}.
\begin{table}[H]
\caption{Classification of $\mathbf{C}^{r_1}\mathbf{C}^{r_2}\mathbf{C}^{r_3}$ trajectories}
\begin{center}
\label{table: tangency value}
\begin{tabular}{ |p{1.66cm}|p{1.205cm}|p{1.205cm}|p{1.205cm}|p{1.205cm}|}
 \hline
  & \multicolumn{2}{|c|}{$k\in[-\pi/2,\pi/2)$} & \multicolumn{2}{|c|}{$k\in[\pi/2,3\pi/2)$}  \\
 \cline{2-5}
 &$|r_1|\geq|r_3|$ &$|r_1|<|r_3|$&$|r_1|\geq|r_3|$&$|r_1|<|r_3|$\\
  \hline
 $r_1>0$,~$r_3>0$& \centering $LRL$& \centering  $LLL$& \centering  $LLL$&   $LRL$ \\
  \hline
 $r_1<0$,~$r_3<0$&   \centering $RLR$& \centering   $RRR$&   \centering $RRR$& $RLR$\\
  \hline
 $r_1>0$,~$r_3<0$& \multicolumn{2}{|c|}{$LRR$}&\multicolumn{2}{|c|}{$LLR$}\\
 \hline
 $r_1<0$,~$r_3>0$&\multicolumn{2}{|c|}{$RLL$}&\multicolumn{2}{|c|}{$RRL$}\\
 \hline
\end{tabular}
\end{center}
\end{table}
\begin{remark}
\label{remark: eight types for fixed magnitude}
Consider the case where, for any given $A$ and $B$, only the radii of $\mathcal{C}_1$ and $\mathcal{C}_2$ are specified. In other words, we only have the magnitudes of $r_1$ and $r_3$. In such a scenario, it follows from Table \ref{table: tangency value} that all eight types of the trajectories ($LLL,RRR,LLR,RRL,LRR,RLL,LRL,RLR$) can be constructed simply by choosing suitable signs of $r_1$ and
$r_3$.
\end{remark}

Remark \ref{remark: eight types for fixed magnitude} highlights the existence of a trajectory of each type. The corresponding value of $r_2\in\mathbb{R}$ for each type can be easily computed. Given any $r_1$ and $r_2$, eqn. \eqref{eq: o2 centre of c2} gives the coordinates of the center of $\mathcal{C}_2$ and Table \ref{table: tangency value} states the kind of tangency between $\mathcal{C}_1$ and $\mathcal{C}_2$ for each $k$. Using eqns. \eqref{eq:centres o1 and o3} and \eqref{eq:centre disctance}, we get the value of $r_2$ with appropriate signs. We derive its expression in a case-wise manner and summarise it in Table \ref{table: r_2 value}. Note that we define $s=d(\mathbf{o_2},\mathbf{o_1})$.
\begin{table}[h]
\caption{Analytical expressions of $r_2$}
\begin{center}
\label{table: r_2 value}
\begin{tabular}{ |p{1.66cm}|p{1.205cm}|p{1.205cm}|p{1.205cm}|p{1.205cm}|}
 \hline
  & \multicolumn{2}{|c|}{$k\in[-\pi/2,\pi/2)$} & \multicolumn{2}{|c|}{$k\in[\pi/2,3\pi/2)$}  \\
 \cline{2-5}
 &$|r_1|\geq|r_3|$ &$|r_1|<|r_3|$&$|r_1|\geq|r_3|$&$|r_1|<|r_3|$\\
  \hline
 $r_1>0$,~$r_3>0$& $-s+r_1$& $s+r_1$& $s+r_1$& $-s+r_1$\\
  \hline
 $r_1<0$,~$r_3<0$& $s+r_1$& $-s+r_1$& $-s+r_1$& $s+r_1$\\
  \hline
 $r_1>0$,~$r_3<0$& \multicolumn{2}{|c|}{$-s+r_1$}&\multicolumn{2}{|c|}{$s+r_1$}\\
 \hline
 $r_1<0$,~$r_3>0$&\multicolumn{2}{|c|}{$s+r_1$}&\multicolumn{2}{|c|}{$-s+r_1$}\\
 \hline
\end{tabular}
\end{center}
\end{table}

The above table can be alternatively expressed compactly in the following form.
\begin{align}
    r_2=\begin{dcases}
        \text{sign}(r_1(|r_1|-|r_3|)(k-\frac{\pi}{2}))s+r_1&, r_1r_2>0\\
        \text{sign}(r_1(k-\frac{\pi}{2}))s+r_1&, r_1r_2<0
    \end{dcases}
    \end{align}
where
\begin{equation*}
    \text{sign}(x)=\begin{dcases}
        ~~1, &x\geq0\\
        -1, &x<0
    \end{dcases}
\end{equation*}
With the radius of the $\mathcal{C}_2$ appropriately defined, the analytical expression for the two changeover points $\{\mathbf{c_1,c_2}\}$ can be written as
\begin{subequations}
    \begin{align}
        \mathbf{c_1}&=\frac{r_2\mathbf{o_1}-r_1\mathbf{o_2}}{r_2-r_1}\\
        \mathbf{c_2}&=\frac{r_2\mathbf{o_3}-r_3\mathbf{o_2}}{r_2-r_3}
        \label{eq: changeover points c_1 c_2}
    \end{align}
\end{subequations} 
Finally, we put forth an interesting observation on $\mathcal{C}_2$.
%\begin{corollary}
 %   The smallest radius, denoted by $|r_2|$, of circle $\mathcal{C}_2$ exists at $k=0$ or $k=\pi$ if $\mathcal{C}_1\cap\mathcal{C}_3=\phi$ and is zero otherwise.
%\end{corollary}
\begin{corollary}
     For $k=\pm\frac{\pi}{2}$, $\mathcal{C}_2$ limits to a straight line (a circle with infinite radius) tangent to both $\mathcal{C}_1$ and $\mathcal{C}_3$.
     \label{corr: r_2 infinity}
\end{corollary}
\begin{proof}
    For $k$ tending to $\pm\pi/2$, it follows from Table \ref{table: r_2 value} that $s$ and, equivalently, $|r_2|$ tend to $\infty$. Additionally, we know from Theorem \ref{thm: hyperbola of centre}, $\mathcal{C}_2$ is always tangent to the other two circles. Thus, $\mathcal{C}_2$ limits to a straight line tangent to $\mathcal{C}_1$ and $\mathcal{C}_3$. Hence, proved.
\end{proof}

Note that Corollary \ref{corr: r_2 infinity} states that as $k$ tends $\pm\pi/2$, the $\mathbf{C}^{r_1}\mathbf{C}^{r_2}\mathbf{C}^{r_3}$ trajectory limits to a CSC trajectory ($C^{r_1}C^{\pm\infty}C^{r_3}$ trajectory). Thus, the proposed trajectory design encapsulates all of the forms of the Dubins Shortest Paths for appropriate variation of the three circles.

 The analysis of $\mathbf{C}^{r_1}\mathbf{C}^{r_2}\mathbf{C}^{r_3}$ trajectories presented in this section is without any explicit constraints of a desired length of the trajectory or curvature boundedness. In the following sections, we proceed to analyse the variation of the length of the trajectory with parameters $\{r_1,r_3,k\}$ and derive the set of reachable lengths using the proposed trajectory between any pair of oriented points.
\section{Variation of $k$ for Circle-Circle-Circle trajectories}
Every attribute of a Circle-Circle-Circle trajectory (like $\mathbf{c_1,c_2},r_2$, \textit{etc.}) can written as a function of the parameter $k$ for fixed values of $r_1$ and $r_3$. We define a function $l:[-\pi/2,3\pi/2)\longrightarrow\mathbb{R}^+$ as the length of trajectory for fixed values of $r_1$ and $r_3$. The parameter $k$ varies in the range $[-\pi/2,3\pi/2)$. This variation of $k$ can be further divided into two parts: $k\in[-\pi/2,\pi/2)$ and $k\in[\pi/2,3\pi/2)$, i.e., over each branch of $\mathcal{H}$. The following interesting results then arise: 
\begin{itemize}
    \item Such a variation of $k$ results in distinct types of $\mathbf{C}^{r_1}\mathbf{C}^{r_2}\mathbf{C}^{r_3}$ trajectories for each branch of $\mathcal{H}$ (which is discussed elaborately in Sec. \ref{subs: classification of CCC}).
    \item An infinite elongation of the trajectory can be achieved within each branch of $\mathcal{H}$. 
\end{itemize}
 We analyze the impact of these variations on the changeover points and the length of trajectory in a case-wise manner. 
\label{section: traj of desired length}

\textbf{Case 1: } $r_1r_3>0$

We know that for both $k=\pi/2$ and $k=-\pi/2$, $\mathcal{C}_2$ becomes a common tangent line to $\mathcal{C}_1$ and $\mathcal{C}_3$. These tangents, as shown in Fig. \ref{fig: partition of circle 1}, divide the boundary of each circle into two parts. We label them as $\mathcal{B}_1$ and $\mathcal{B}_2$ as shown in Fig. \ref{fig: partition of circle 1}. One of the variations in $k$ results in $\mathcal{C}_2$ being externally tangent to both $\mathcal{C}_1$ and $\mathcal{C}_3$. Under such a variation, the changeover points lie in $\mathcal{B}_2$ for both $\mathcal{C}_1$ and $\mathcal{C}_3$ as shown in Fig. \ref{fig: case 1 ext}. 
\begin{figure}[h]
    \centering
    \includegraphics[scale=0.39]{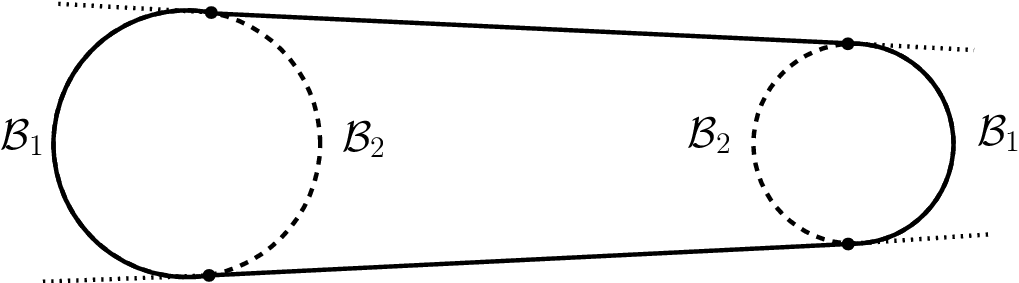}
    \caption{Common tangents for $\mathcal{C}_1$ and $\mathcal{C}_3$ for $r_Ar_B>0$}
    \label{fig: partition of circle 1}
\end{figure}
The other variation results in $\mathcal{C}_2$ being internally tangent to both $\mathcal{C}_1$ and $\mathcal{C}_3$. The changeover points lie in $\mathcal{B}_1$ for both $\mathcal{C}_1$ and $\mathcal{C}_3$ in such a case as shown in Fig. \ref{fig: case 1 int}. For both of the variations, the lengths of the $\mathbf{C}^{r_1}\mathbf{C}^{r_2}\mathbf{C}^{r_3}$ trajectories tend to infinity at one end of each branch of $\mathcal{H}$.

\begin{figure*}[htbp]
     \centering
     \begin{subfigure}[b]{0.25\textwidth}
         \centering
         \includegraphics[width=\textwidth]{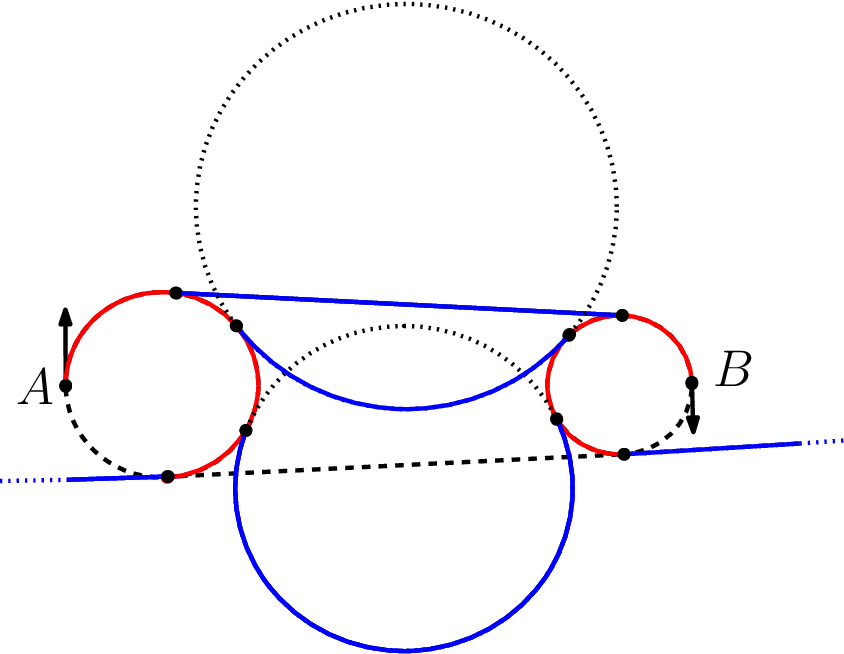}
         \caption{$\mathcal{C}_2$ externally tangent to both $\mathcal{C}_1$ and $\mathcal{C}_3$}
         \label{fig: case 1 ext}
     \end{subfigure}
     \hfill
     \begin{subfigure}[b]{0.23\textwidth}
         \centering
         \includegraphics[width=0.85\textwidth]{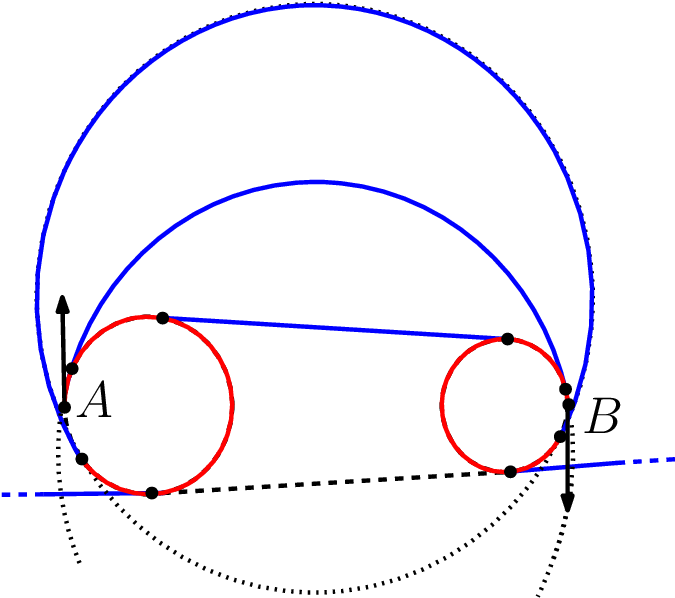}
         \caption{$\mathcal{C}_2$ internally tangent to both $\mathcal{C}_1$ and $\mathcal{C}_3$}
         \label{fig: case 1 int}
     \end{subfigure}
     \hfill
     \begin{subfigure}[b]{0.24\textwidth}
         \centering
         \includegraphics[width=\textwidth]{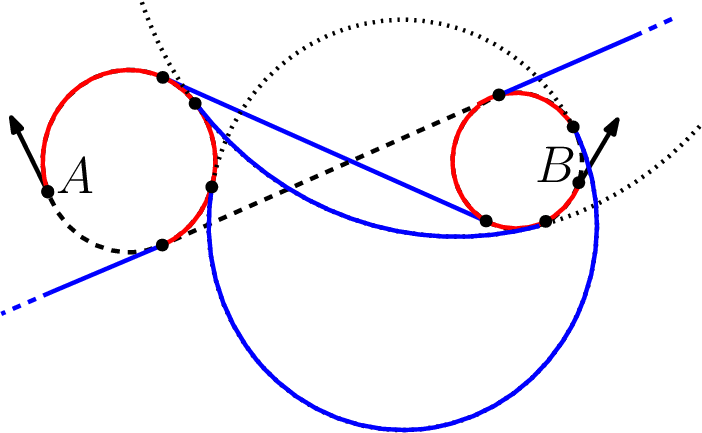}
         \caption{$\mathcal{C}_2$ externally tangent to $\mathcal{C}_1$ and internally tangent to $\mathcal{C}_3$}
         \label{fig: case 2_1}
     \end{subfigure}
     \hfill
     \begin{subfigure}[b]{0.24\textwidth}
         \centering
         \includegraphics[width=\textwidth]{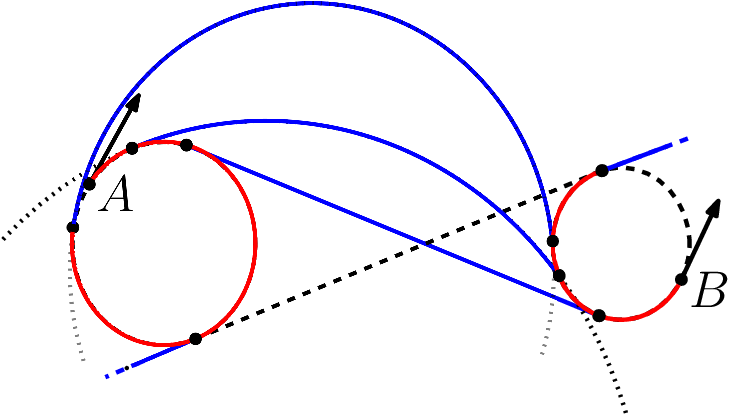}
         \caption{$\mathcal{C}_2$ internally tangent to $\mathcal{C}_1$ and externally tangent to $\mathcal{C}_3$}
         \label{fig: case 2_2}
     \end{subfigure}
        \caption{Variation of $\mathcal{C}_2$ in a $\mathbf{C}^{r_1}\mathbf{C}^{r_2}\mathbf{C}^{r_3}$ trajectory}
        \label{fig:three graphs}
\end{figure*}

\textbf{Case 2:} $r_1r_3<0$
 
 In contrast to the previous case, the tangents, arising out $k=\pi/2$ and $k=-\pi/2$, take the form shown in Fig. \ref{fig: partition of circle 2}. As before, they divide the boundary of each circle into two parts: $\mathcal{B}_1$ and $\mathcal{B}_2$. The variation of $k$ in $[-\pi/2,\pi/2)$ results in $\mathcal{C}_2$ being externally tangent to $\mathcal{C}_1$ and internally tangent to $\mathcal{C}_3$. Unlike the previous case, the changeover points lie in $\mathcal{B}_2$ for $\mathcal{C}_1$ and in $\mathcal{B}_2$ for $\mathcal{C}_3$ under such a variation as shown in Fig. \ref{fig: case 2_1}. The other variation results in $\mathcal{C}_2$ being internally tangent to $\mathcal{C}_1$ and externally tangent to $\mathcal{C}_3$ as shown in Fig. \ref{fig: case 2_2}. For both of the variations, the lengths of the $\mathbf{C}^{r_1}\mathbf{C}^{r_2}\mathbf{C}^{r_3}$ trajectories tend to infinity at one end of each branch of $\mathcal{H}$.
\begin{figure}[H]
    \centering
    \includegraphics[scale=0.39]{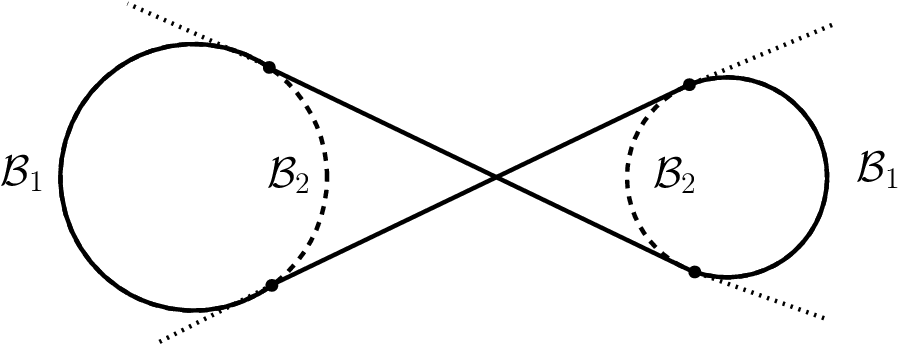}
    \caption{Common tangents for $\mathcal{C}_1$ and $\mathcal{C}_3$ for $r_Ar_B<0$}
    \label{fig: partition of circle 2}
\end{figure}
Note that while $l(k)$ goes to infinity for each variation of $k$, $l(k)$ is not always continuous. This is illustrated through the following result.
\begin{lemma}
    Consider two oriented points $A$ and $B$. For fixed values of $r_1,r_3\in\mathbb{R}$, $l(\cdot)$ has at most two discontinuities. The points of discontinuity, if they exist, are at $k_a$ and $k_b$ where $\mathbf{c_1}(k_a)=\mathbf{a}$ and $\mathbf{c_2}(k_b)=\mathbf{b}$. Moreover, these discontinuities are jump discontinuities of magnitude $2\pi |r_1|$ at $k=k_a$ and $2\pi |r_3|$ at $k=k_b$.
    \label{lem: jump discontinuity}
\end{lemma}
\begin{proof}
    For each of the variations of $k$, we discussed in the preceding paragraphs how $\mathbf{c_1}$ and $\mathbf{c_2}$, i.e., the changeover points, vary continuously in either $\mathcal{B}_1$ or $\mathcal{B}_2$. Let $k_a$ and $k_b$ be values such that $\mathbf{c_1}(k_a)=\mathbf{a}$ and $\mathbf{c_2}(k_b)=\mathbf{b}$. Consider an infinitesimal variation of $k$ around the point $k=k_a$ and the resulting trajectories as shown in Fig. \ref{fig: discontinuity of l}. 
\begin{figure}[H]
    \centering
    \includegraphics[scale=0.55]{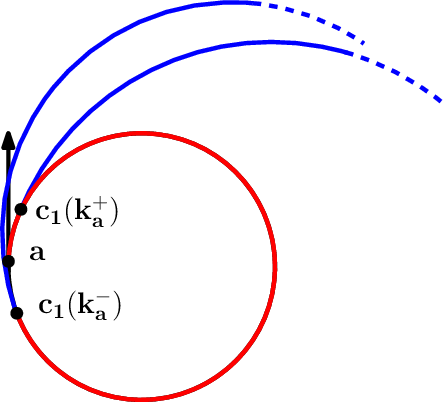}
    \caption{Jump discontinuity in $l(\cdot)$}
    \label{fig: discontinuity of l}
\end{figure}
It is easy to see in Fig. \ref{fig:traj_cc} that the arc length of $\widearc{\mathbf{c_1c_2}}$ and $\widearc{\mathbf{c_2b}}$ is continuous along this infinitesimal variation around $k_a$. However, the arc length of $\widearc{\mathbf{ac_1}}$ tends to $0$ as $k_a^+$ tends to $k_a$ and $\widearc{\mathbf{ac}_1}$ tends to $2\pi |r_1|$ as $k_a^-$ tends to $k_a$. Thus, $l(\cdot)$ is discontinuous at $k=k_a$. Note that in the above case, $l(k)$ jumps down by $2\pi |r_1|$ at $k=k_a$. If the location of  $\mathbf{c_1}(k_a^+)$ and $\mathbf{c_1}(k_a^-)$ were switched, $l(k)$ would have jumped up by $2\pi|r_1|$. A similar discontinuity exists at $k=k_b$. Hence, proved.
\end{proof}

It is important to note that the discontinuities at $k_a$ and $k_b$ need not occur simultaneously. Further, if $\mathbf{c_1}(k)$ is not equal to $\mathbf{a}$ and $\mathbf{c_b}(k)$ is not equal to $\mathbf{b}$ for any $k$ varying over one branch of the hyperbola, we get a continuous elongation of $l(k)$ to infinity. This is true only if no curvature constraints are imposed on the trajectory. \textit{Through these results, we observe that the variation of $k$ results in interesting properties pertaining to the length of the trajectory.} We use these results along with the imposition of curvature constraint to determine the set reachable lengths between any two oriented points in the following section.

\section{Curvature-bounded Circle-Circle-Circle trajectory of desired length}
\label{sec: reachability set}
With important observations arising out of the variation of parameter $k$ as discussed in the previous section, we are now able to state the set of reachable lengths. We begin the analysis with a reference to \cite{chen_elongation} for a characterisation of pairs of oriented points $(A,B)$ based upon the kind of Dubins Shortest Path ($\Lambda_m$) between them:
\begin{subequations}
\label{eq: classification of O}
\begin{align}
        \mathcal{O}_1&=\{(A,B) | \Lambda_m\in C_{\eta}S_dC_{\zeta}\text{ with } \eta\geq\pi\}\\
        \mathcal{O}_2&=\{(A,B) | \Lambda_m\in C_{\eta}S_dC_{\zeta}\text{ with } \zeta\geq\pi\}\\
        \mathcal{O}_3&=\{(A,B) | \Lambda_m\in C_{\eta}S_dC_{\zeta}\text{ with } d\geq4r_{\min}\}\\
        \mathcal{O}_4&=\{(A,B) | \Lambda_m\in C_{\eta}S_dC_{\zeta}\text{ with } \newline d(\mathbf{c}^r_A,\mathbf{c}^r_B)\geq4r_{\min}\}\\
        \mathcal{O}_5&=\{(A,B) | \Lambda_m\in C_{\eta}S_dC_{\zeta}\text{ with } \newline d(\mathbf{c}^l_A,\mathbf{c}^l_B)\geq4r_{\min}\}
\end{align}
\end{subequations}
where $\eta$ and $\zeta$ are the arc lengths of the first and third circles, respectively, and $d$ is the length of the straight line path. Let $\mathcal{O}:=\mathcal{O}_1\cup\mathcal{O}_2\cup\mathcal{O}_3\cup\mathcal{O}_4\cup\mathcal{O}_5$ and $\mathcal{O}^c$ be the its complementary set
\begin{equation*}
    \mathcal{O}^c=\{(A,B) | \Lambda_m\in CSC, (A,B)\notin \mathcal{O}\}.
\end{equation*}

Based upon the above characterisation, the authors in \cite{chen_elongation} classify the set of reachable lengths for any two oriented points as:
\begin{theorem}[\cite{chen_elongation} ]
\label{thm: maximal reachable lengths}
    Given any two oriented points $A$ and $B$ so that $A\neq B$, the following statements hold:
    \begin{itemize}
        \item[1.] If $(A,B)\notin\mathcal{O}\cup\mathcal{O}^c$, for every $l_o\geq l_m$ there exists a trajectory $\Lambda$ so that $l(\Lambda)=l_o$.
        \item[2.] If $(A,B)\in\mathcal{O}$, for every $l_o\geq l_m$ there exists a trajectory $\Lambda$ so that $l(\Lambda)=l_o$. 
        \item[3.] If $(A,B)\in\mathcal{O}^c$, we have that (a) for every $l_o\in[l_m,l_1]\cup[l_2,\infty)$ there exists a trajectory $\Lambda$ so that $l(\Lambda)=l_o$; and (b) for any trajectory $\Lambda$ we have $l(\Lambda)\notin(l_1,l_2)$.
    \end{itemize}
\end{theorem}
where $l_o$ is the length of the desired trajectory and
\begin{subequations}
\label{eq: l1 l2 definition}
\begin{align}
    l_1&:=\max\{l^s_{LRL},l^s_{RLR}\}\\
    l_2&:=\min\{l_m+2\pi r_{\min},l^l_{LRL},l^l_{RLR},\{l_{RSR},l_{LSL},l_{RSL},l_{LSR}\}\setminus\{l_m\}\}
\end{align}
\end{subequations}

\begin{figure}[h]
    \centering
    \includegraphics[width=0.5\linewidth]{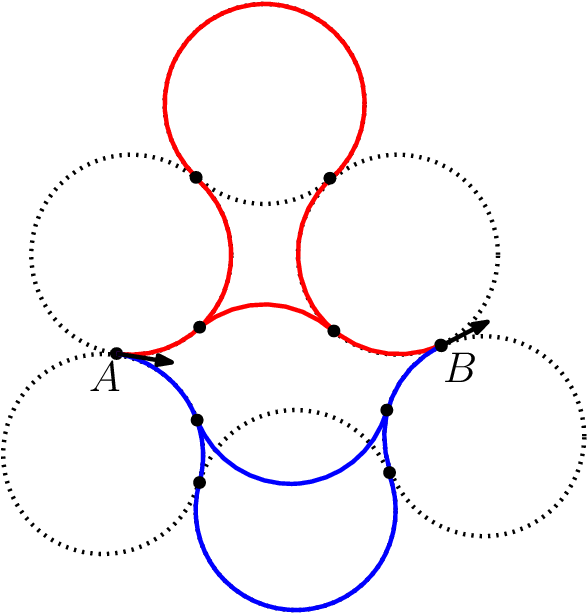}
    \caption{Two $LRL$ trajectories (in red) and $RLR$ trajectories (in blue) between $A$ and $B$ }
    \label{fig: lrl rlr 2}
\end{figure}

The various lengths of trajectory mentioned in \eqref{eq: l1 l2 definition} have been elaborated in \cite{chen_elongation}. For $(A,B)\in\mathcal{O}^c$, there exist two $LRL$ and $RLR$ trajectories with the magnitude of curvature being $1/|r_{\min}|$ throughout as shown in Fig. \ref{fig: lrl rlr 2} in red and blue, respectively. Note that there exists a continuous elongation between the two $LRL$ (or $RLR$) trajectories if there were no curvature constraints. We denote the length of the longer $LRL$ trajectory as $l_{LRL}^l$ and the shorter one as $l_{LRL}^s$. The lengths $l_{RLR}^l$ and $l_{RLR}^s$ are defined similarly. These lengths are used to determine $l_1$ and $l_2$ in eqn. \eqref{eq: l1 l2 definition}. Further, the following result holds for the remaining of the trajectories in eqn. \eqref{eq: l1 l2 definition} as shown in \cite{chen_elongation}.

\begin{lemma}[\cite{chen_elongation}]
    If $(A,B)\in\mathcal{O}^c$, each CSC-path, with its length in $\{l_{RSR},l_{LSL},l_{RSL},l_{LSR}\}\setminus\{l_m\}$, has anti-parallel tangents.
    \label{lem: anti parallel tangent}
\end{lemma}
Lemma \ref{lem: anti parallel tangent} highlights that for each CSC-path, with its length in $\{l_{RSR},l_{LSL},l_{RSL},l_{LSR}\}\setminus\{l_m\}$, has either $\eta\geq\pi$ or $\zeta\geq\pi$. With these observations, we present the elongation strategies for the Dubins Shortest Path to achieve a $\mathbf{C}^{r_1}\mathbf{C}^{r_2}\mathbf{C}^{r_3}$ trajectory of a desired length between any two oriented points.
\subsection{Elongation of a trajectory for \texorpdfstring{$(A,B)\notin\mathcal{O}\cup\mathcal{O}^c$}{Lg}}
The set of ordered pairs $(A,B)\notin\mathcal{O}\cup\mathcal{O}^c$ corresponds to the cases where the Dubins Shortest Path is a CCC trajectory. The following theorem illustrates the elongation strategy for such cases.
\begin{theorem}
    Given oriented points $A$ and $B$ such that $(A,B)\notin\mathcal{O}\cup\mathcal{O}^c$, there always exists a $\mathbf{C}^{r_1}\mathbf{C}^{r_2}\mathbf{C}^{r_3}$ trajectory of any desired length $l_o\in[l_m,\infty)$.
    \label{lem: CCC ELONGATION}
\end{theorem}
\begin{proof}
\begin{figure}[h]
        \centering
        \includegraphics[scale=0.4]{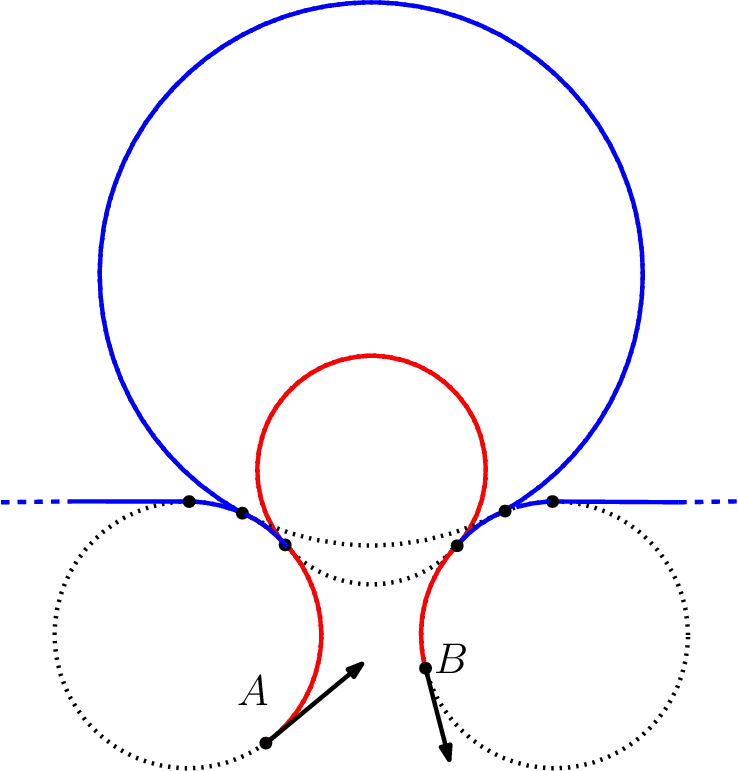}
        \caption{Elongation of CCC Dubins Shortest Path}
        \label{fig:ccc elongation}
    \end{figure}
    It is shown in \cite{dubins} that if the Dubins Shortest Path is of the form CCC, the arc length of the middle circle is greater than or equal to $\pi r_{\min}$. Thus, the minimum length trajectory is shown in red in Fig. \ref{fig:ccc elongation} (for $LRL$). If we increase the magnitude of $r_2$, we get a continuous elongation till infinity. We can elongate the Dubins Shortest Path of form $RLR$ similarly.
    Hence, proved.
\end{proof}
\subsection{Elongation of a trajectory for  \texorpdfstring{$(A,B)\in\mathcal{O}$}{Lg}}
The set $\mathcal{O}$ is formed by the union of five different sets given by eqn. \eqref{eq: classification of O}. We first present an important observation on this set.
\begin{lemma}
    For all $(A,B)\in\mathcal{O}_3$, $(A,B)\in\mathcal{O}_1\cup\mathcal{O}_2\cup\mathcal{O}_4\cup\mathcal{O}_5$.
    \label{lem: not in o3}
\end{lemma}
\begin{proof}
Consider the the Dubins Shortest Path between any $(A,B)\in\mathcal{O}_3$ as a $C_{\eta}S_dC_{\zeta}$ trajectory. Clearly, if $\eta\geq\pi$ or $\zeta\geq\pi$, $(A,B)\in\mathcal{O}_1\cup\mathcal{O}_2$. We now focus on the case of $\eta<\pi$ and $\zeta<\pi$. If the trajectory is of form $RSR$, $d(\mathbf{c}^r_A,\mathbf{c}^r_B)=d\geq4r_{\min}$ implying $(A,B)\in\mathcal{O}_4$. Similarly, $(A,B)\in\mathcal{O}_5$ if the trajectory is of the form $LSL$.

If the trajectory is of the form $RSL$, the same has been shown in Fig. \ref{fig: rsl o3 not}. Construct the normal lines at the endpoints of the straight line segment (labelled as $n_1$ and $n_2$) and circle $\mathcal{C}_A^l$. As $\eta<\pi$, the point $\mathbf{c}^l_A$ lies to the left of or on the normal line $n_1$. Clearly, $d(\mathbf{c}^l_A,\mathbf{c}^l_B)\geq d\geq4r_{\min}$. Thus, $(A,B)\in\mathcal{O}_5$. Note that we can similarly show that $d(\mathbf{c}^r_A,\mathbf{c}^r_B)\geq d\geq4r_{\min}$ and $(A,B)\in\mathcal{O}_4$ for this case. Further, the case of $LSR$ path can be proven similarly. Hence, proved.
\begin{figure}[h]
    \centering
    \includegraphics[scale=0.6]{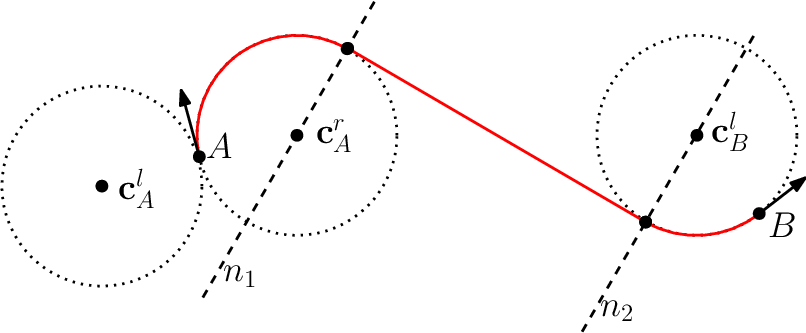}
    \caption{For a Dubins Shortest Path of form $RSL$ $d(\mathbf{c}^l_A,\mathbf{c}^l_B)\geq4r_{\min}$}
    \label{fig: rsl o3 not}
\end{figure}
\end{proof}

The sets given in eqn. \eqref{eq: classification of O} are not mutually disjoint. There exist certain pairs of $(A,B)$ that belong to a unique set and certain pairs that belong to multiple sets. Fig. \ref{fig: belong all} illustrates a pair of oriented point that belongs to $\mathcal{O}_1\cap\mathcal{O}_4\cap\mathcal{O}_5$. However, Lemma \ref{lem: not in o3} illustrates that $\mathcal{O}_3$ does not have any unique elements. Thus, it is sufficient to construct elongation strategies for the remaining four sets to explore a continuous elongation in $\mathcal{O}$. We analyse each of these sets separately. We begin by presenting the following result which discusses elongation strategies for set $\mathcal{O}_1\cup\mathcal{O}_2$.

\begin{figure}[h]
    \centering
    \includegraphics[scale=0.6]{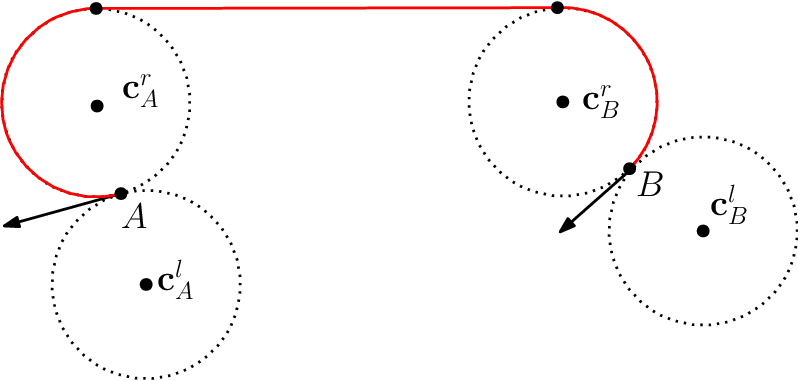}
    \caption{CSC Dubins Shortest Path for some $(A,B)\in\mathcal{O}_1\cap\mathcal{O}_4\cap\mathcal{O}_5$}
    \label{fig: belong all}
\end{figure}

\begin{lemma}
    Given oriented points $A$ and $B$ such that $(A,B)\in\mathcal{O}_1\cup\mathcal{O}_2$, there always exists $\mathbf{C}^{r_1}\mathbf{C}^{r_2}\mathbf{C}^{r_3}$ trajectory of any desired length $l_o\in[l_m,\infty)$.
    \label{lem: CSC great pi ELONGATION}
\end{lemma}
\begin{proof}
We begin the proof by considering the set $\mathcal{O}_1$. Thus, $\eta\geq\pi$. We divide the proof into two cases.

\textbf{Case 1}: Let the $CSC$ trajectory be of form $RSR$. Clearly, $\mathbf{a}\in\mathcal{B}_2$. We continuously deform $\mathcal{C}_2$ as shown in Fig. \ref{fig: case 1 int}. If $\mathbf{b}\in\mathcal{B}_2$, this elongation will be continuous till infinity. 

On the contrary, if $\mathbf{b}\in\mathcal{B}_1$, there exists a $k_b$ such that $\mathbf{c_2}(k_b)=\mathbf{b}$. Construct circle $\mathcal{C}^l_B$ and a transverse tangent between $\mathcal{C}_1$ and $\mathcal{C}^l_B$.  This construction is shown in Fig. \ref{fig: rsr 1}. The point $\mathbf{x}$ divides $\mathcal{B}_2$ into two arcs denoted by $\mathcal{B}_2^I$ and $\mathcal{B}_2^{II}$. If $\mathbf{a}\in\mathcal{B}_2^I$, we deform $\mathcal{C}_2$ and get a series $RRR$ trajectories until point $k_b$ is reached as shown by a green curve in Fig. \ref{fig: rsr 1}. This is an $RR$ trajectory. If we continue further, we will have a jump of $2\pi|r_3|$ as discussed in Lemma \ref{lem: jump discontinuity}. To ascertain a continuous elongation of the length of the trajectory, we proceed to deform $\mathcal{C}_2$ such that it is internally tangent to $\mathcal{C}_1$ and externally to $\mathcal{C}_B^l$ as shown in blue in Fig. \ref{fig: rsr 1} resulting in a $RRL$ trajectory. Since $\mathbf{a}\in\mathcal{B}_2^{I}$, we will have a continuous elongation until $|r_2|$ goes to infinity as shown in Fig. \ref{fig: rsr 1}.

\begin{figure}[h]
     \centering
     \begin{subfigure}[b]{0.33\textwidth}
         \centering
         \includegraphics[width=\textwidth]{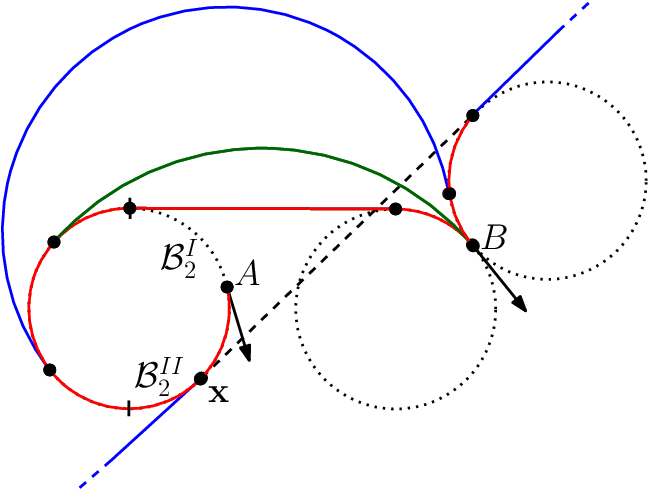}
         \caption{$\mathbf{a}\in\mathcal{B}_2^{I}$ and $\mathbf{b}\in\mathcal{B}_1$}
         \label{fig: rsr 1}
     \end{subfigure}
     \hfill
     \begin{subfigure}[b]{0.39\textwidth}
         \centering
         \includegraphics[width=\textwidth]{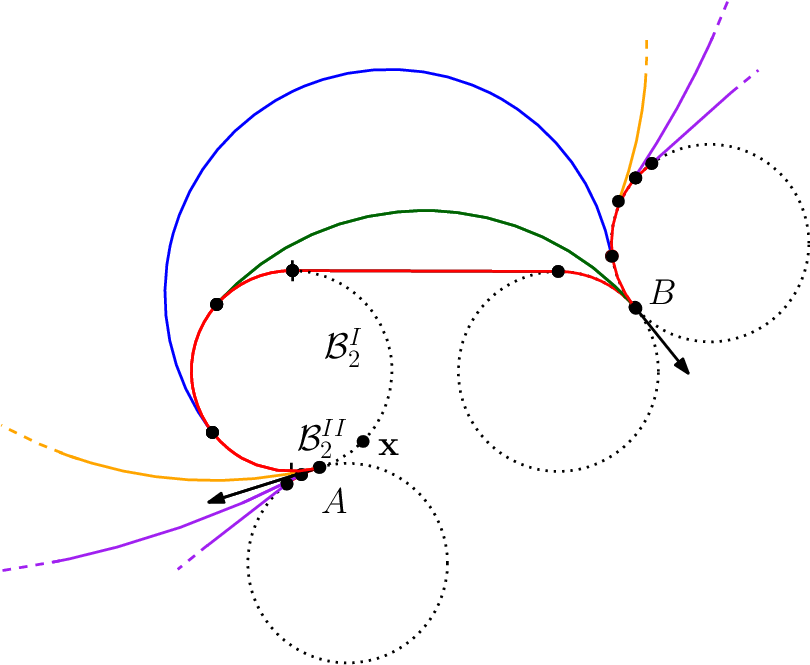}
         \caption{$\mathbf{a}\in\mathcal{B}_2^{II}$ and $\mathbf{b}\in\mathcal{B}_1$}
         \label{fig: rsr 2}
     \end{subfigure}
        \caption{Elongation of $C_{\eta}S_dC_{\zeta}$ trajectory with $\eta\geq\pi$ }
        \label{fig: rsr}
\end{figure}

If $\mathbf{a}\in\mathcal{B}_2^{II}$, construct $\mathcal{C}^l_A$. We deform $\mathcal{C}_2$ similarly as in the previous case until the point $k_a$ is reached such that $\mathbf{c_1}(k_a)=a$ before the length of trajectory goes to infinity. This is an $RL$ trajectory as shown by the yellow curve in Fig. \ref{fig: rsr 2}. If we proceed similarly, we get a jump discontinuity of $2\pi|r_1|$. Instead, we now deform $\mathcal{C}_2$ such that it is externally tangent to both $\mathcal{C}^l_A$ and $\mathcal{C}^l_B$ resulting in $LRL$ trajectories and a continuous deformation after $k=k_a$ until the length reaches infinity.
Note that in all the cases, $|r_2|\geq r_{\min}$. Thus, the curvature constraint on the trajectory is not violated. We can similarly achieve a continuous elongation of an $LSL$ path with $\eta\geq\pi$.

\textbf{Case 2}: Let the $CSC$ trajectory be of form $RSL$. For $\eta\geq\pi$, $\mathbf{a}$ can lie in either $\mathcal{B}_1$ or $\mathcal{B}_2$. If $\mathbf{a}\in\mathcal{B}_2$, we deform $\mathcal{C}_2$ as shown in Fig. \ref{fig: case 2_2} such that $\mathbf{c_1}\neq\mathbf{a}$ always. Further, if $\mathbf{b}\in\mathcal{B}_1$, we get a continuous elongation as $\mathbf{c_2}\in\mathcal{B}_2$ and no jump discontinuity is encountered. If $\mathbf{b}\in\mathcal{B}_2$, there exist a $k_b$ such that $\mathbf{c_2}(k_b)=\mathbf{b}$. However, there is a jump down discontinuity and the set of reachable lengths is given by $l\in[l_m,l_b+2\pi|r_3|)\cup[l_b,\infty)=[l_m,\infty)$ where $l(k_b)=l_b$. Thus, a continuous elongation of the trajectory exists until its length goes to infinity.
\begin{figure}[h]
    \centering
    \includegraphics[scale=0.4]{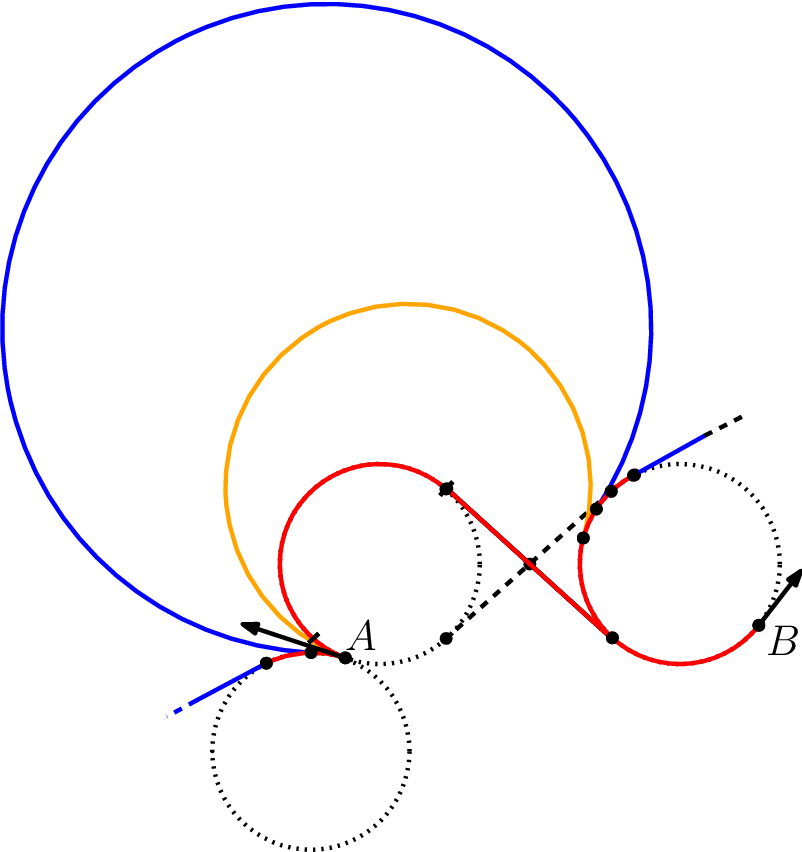}
    \caption{Elongation of $RSL$ trajectory with $\eta\geq\pi$ where $\mathbf{a}\in\mathcal{B}_1$ and $\mathbf{b}\in\mathcal{B}_1$}
    \label{fig: lsr elong}
\end{figure}

If $\mathbf{a}\in\mathcal{B}_1$, we again proceed in a case-wise manner. If $\mathbf{b}\in\mathcal{B}_2$, we deform $\mathcal{C}_2$ as shown in Fig. \ref{fig: case 2_1}. For such a variation, $\mathbf{c_1}\in\mathcal{B}_2$ and  $\mathbf{c_2}\in\mathcal{B}_1$ resulting in no jump discontinuities and we get a continuous elongation of the trajectory till infinity. If $\mathbf{b}\in\mathcal{B}_2$, we proceed in a similar manner as the case of $\mathbf{a}\in\mathcal{B}_2^{II}$ in Case 1. The same has been highlighted in Fig. \ref{fig: lsr elong}. We can similarly achieve a continuous elongation for an $LSR$ path.

The elongation strategies mentioned here can be similarly applied to the cases where $\zeta\geq\pi$, i.e., $(A,B)\in\mathcal{O}_2$. Thus, a continuous elongation of trajectory exist using $\mathbf{C}^{r_1}\mathbf{C}^{r_2}\mathbf{C}^{r_3}$ trajectory for all $(A,B)\in\mathcal{O}_1\cup\mathcal{O}_2$. Hence, proved.
\end{proof}

 An important idea emerges from the proof of Lemma \ref{lem: CSC great pi ELONGATION}. The deformation of $\mathcal{C}_2$ is realised by the variation of the parameter $k$. If the point $k_a$ (or $k_b$), defined in Lemma \ref{lem: jump discontinuity}, is reached in this deformation, one should switch $\mathcal{C}_1$ (or $\mathcal{C}_3$) from $\mathcal{C}^l_A$ to $\mathcal{C}^r_A$ or vice-versa (or $\mathcal{C}^l_B$ to $\mathcal{C}^r_B$ or vice-versa) for further deformation of $\mathcal{C}_2$ instead of proceeding on the same circle. The arc length of $\widearc{\mathbf{ac_1}}$ (or $\widearc{\mathbf{c_2b}}$) goes to zero as $k_a$ (or $k_b$) is approached and starts increasing after the switch, preserving its continuity and, in turn, the continuity of the overall length of trajectory (see Fig. \ref{fig: cont elong}). The switching of the terminal circles can be alternatively viewed as changing $r_1$ (or $r_3$) from $r_{\min}$ to $-r_{\min}$ or vice versa.

\begin{figure}[h]
    \centering
    \includegraphics[scale=0.5]{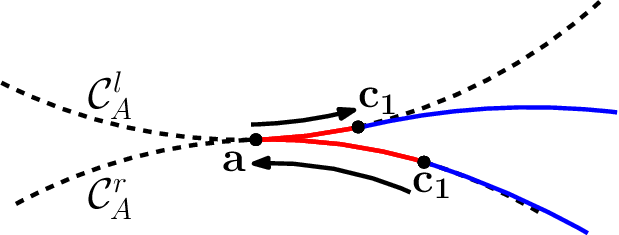}
    \caption{Continuous elongation of $\mathbf{C}^{r_1}\mathbf{C}^{r_2}\mathbf{C}^{r_3}$ trajectory}
    \label{fig: cont elong}
\end{figure}         

We now discuss the set $\mathcal{O}_4\cup\mathcal{O}_5$. Clearly, any pair $(A,B)\in\mathcal{O}_4\cup\mathcal{O}_5$ such that $\eta\geq\pi$ or $\zeta\geq\pi$ can be elongated continuously to infinitely large lengths from Lemma \ref{lem: CSC great pi ELONGATION}. We analyse the remaining cases in the following result.
\begin{lemma}
    Given oriented points $A$ and $B$ such that $(A,B)\in\mathcal{O}_4\cup\mathcal{O}_5$ with $\eta<\pi$ and $\zeta<\pi$, there always exists a $\mathbf{C}^{r_1}\mathbf{C}^{r_2}\mathbf{C}^{r_3}$ trajectory of any desired length $l_o\in[l_m,\infty)$.
    \label{lem: CSC great d(o) ELONGATION}
\end{lemma}

\begin{figure}[h]
     \centering
     \begin{subfigure}[b]{0.42\textwidth}
         \centering
         \includegraphics[width=\textwidth]{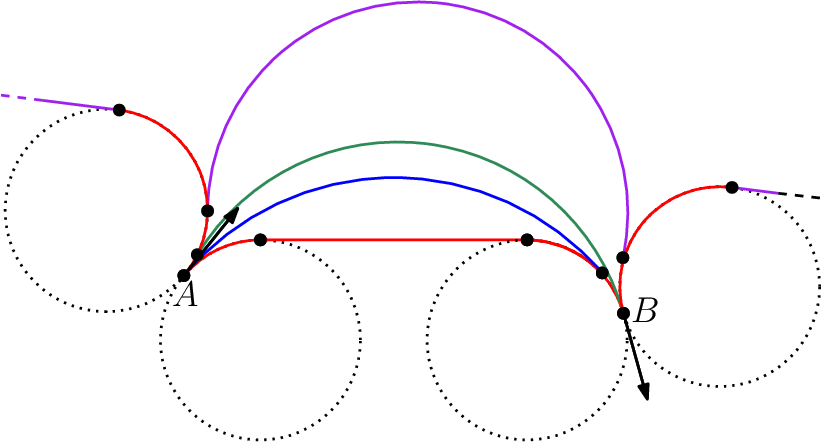}
         \caption{$\Lambda_m$ is of the form $RSR$}
     \end{subfigure}
     \hfill
     \begin{subfigure}[b]{0.36\textwidth}
         \centering
         \includegraphics[width=\textwidth]{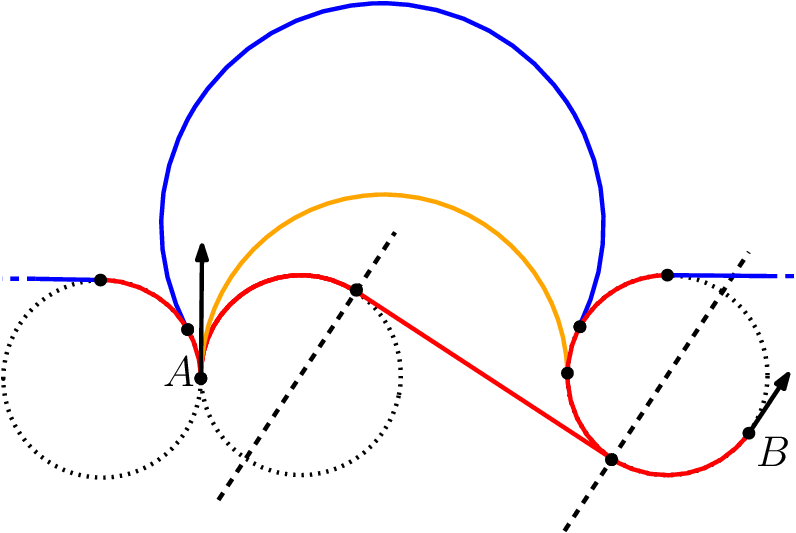}
         \caption{$\Lambda_m$ is of the form $RSL$}
     \end{subfigure}
        \caption{Elongation of $C_{\eta}S_dC_{\zeta}$ trajectory with $d(\mathbf{c}^l_A,\mathbf{c}^l_B)\geq4r_{\min}$}
        \label{fig: O4 O5}
\end{figure}

\begin{proof}
    The proof is similar to that of Lemma \ref{lem: CSC great pi ELONGATION}. It can be illustrated for different cases of $RSR$ and $RSL$ trajectories through Fig. \ref{fig: O4 O5}.
\end{proof}

The following theorem follows from Lemma \ref{lem: CSC great pi ELONGATION} and \ref{lem: CSC great d(o) ELONGATION}.
\begin{theorem}
    Given oriented points $A$ and $B$ such that $(A,B)\in\mathcal{O}$, there exists $\mathbf{C}^{r_1}\mathbf{C}^{r_2}\mathbf{C}^{r_3}$ trajectory for all lengths $l_o\in[l_m,\infty)$.
    \label{lem: set O ELONGATION}
\end{theorem}

\subsection{Elongation of trajectory for \texorpdfstring{$(A,B)\in\mathcal{O}^{c}$}{Lg}}

For all $(A,B)\in\mathcal{O}^c$, there exist trajectories of lengths $l_{LRL}^s$ and $l_{RLR}^s$ as shown in Fig. \ref{fig: lrl rlr 2}. Using the definitions of $l_1$ and $l_2$ from \eqref{eq: l1 l2 definition}, the following result illustrates the elongation strategy for such cases.

\begin{theorem}
    Given oriented points $A$ and $B$ such that $(A,B)\in\mathcal{O}^c$, there always exists a $\mathbf{C}^{r_1}\mathbf{C}^{r_2}\mathbf{C}^{r_3}$ trajectory of any desired length $l_o\in[l_m,l_1]\cup[l_2,\infty)$ where $l_1$ and $l_2$ are given by eqn. \eqref{eq: l1 l2 definition}.
    \label{lem: set OC elongation}
\end{theorem}
\begin{proof}
    Consider the  trajectories whose lengths are mentioned in \eqref{eq: l1 l2 definition}. Without loss of generality, let $l_1=l_{RLR}^s$. Fig. \ref{fig: lsl_oc} shows the elongation strategy for an $LSL$ minimum path. We deform $\mathcal{C}_2$ resulting in a series of $LLL$ trajectories until point $k_a$ is reached. This trajectory is of the type $LL$. To preserve the continuity of elongation, we construct $\mathcal{C}_2$ between $\mathcal{C}^r_A$ and $\mathcal{C}_3$ resulting in a series of $RLL$ trajectories until a point $k_b$ is reached. Henceforth, $\mathcal{C}_2$ is constructed such that it is externally tangent to $\mathcal{C}^r_A$ and $\mathcal{C}^r_B$ forming a series of $RLR$ trajectories. This $\mathcal{C}_2$ is deformed until we get a trajectory whose length is $l_{RLR}^s$. Thus, a $\mathbf{C}^{r_1}\mathbf{C}^{r_2}\mathbf{C}^{r_3}$ trajectory exists for every $l_o\in[l_m,l_1]$. Note that it might happen that $k_b$ is reached earlier than $k_a$. In such a case, $\mathcal{C}_3$ is switched first. The proof follows similarly.

    \begin{figure}[h]
        \centering
        \includegraphics[width=0.65\linewidth]{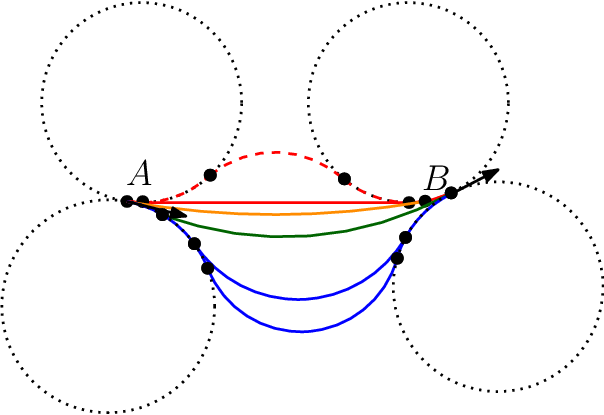}
        \caption{Elongation of trajectory from $l_m$ to $l_1$ when $(A,B)\in\mathcal{C}^c$}
        \label{fig: lsl_oc}
    \end{figure}

    Clearly, trajectories of length $\{l_{RLR}^l,l_{LRL}^l\}$ have arc lengths of $\mathcal{C}_2$ greater than $\pi r_{\min}$ as shown in Fig. \ref{fig: lrl rlr 2}. These can be continuously elongated till infinity using strategies in Theorem \ref{lem: CCC ELONGATION}. Further, we know from Lemma \ref{lem: anti parallel tangent} that the trajectories of lengths $\{l_{RSR},l_{LSL},l_{RSL},l_{LSR}\}\setminus\{l_m\}$ have either $\eta\geq\pi$ or $\zeta\geq\pi$. Thus, they can also be elongated continuously using Lemma \ref{lem: CSC great pi ELONGATION}. Lastly, the trajectory of length $l_m+2\pi r_{\min}$ can also be continuously elongated till infinity using Lemma \ref{lem: CSC great pi ELONGATION}. Consequently, there exists a $\mathbf{C}^{r_1}\mathbf{C}^{r_2}\mathbf{C}^{r_3}$ trajectory for all $l_o\in[l_2,\infty)$. Hence proved.
\end{proof}

Theorems \ref{lem: CCC ELONGATION} - \ref{lem: set OC elongation} show that the proposed $\mathbf{C}^{r_1}\mathbf{C}^{r_2}\mathbf{C}^{r_3}$ trajectories can form curvature-bounded trajectories of any desired lengths. The corresponding set of reachable lengths is exactly the same as the maximum reachability set mentioned in Theorem \ref{thm: maximal reachable lengths} for all pairs of oriented points $(A,B)$. Further, the proposed strategies guarantee a maximum of two changeover points in the entire trajectory between $A$ and $B$ for any desired length. 

The analysis of the set of reachable lengths till now is done with the values of $r_1$ and $r_3$ fixed. They took the value of $\pm r_{\min}$ depending upon the kind of Dubins Shortest Path between $(A,B)$ and the subsequent elongation strategy. We have established that the maximum reachability set is completely covered with just these values of $r_1$ and $r_3$ for the $\mathbf{C}^{r_1}\mathbf{C}^{r_2}\mathbf{C}^{r_3}$ trajectory. Next, we proceed to discuss the variation of $r_1$ and $r_3$ values and its implication on the lengths of the $\mathbf{C}^{r_1}\mathbf{C}^{r_2}\mathbf{C}^{r_3}$ trajectories.

\subsection{Variation of \texorpdfstring{$r_1$}{Lg} and \texorpdfstring{$r_3$}{Lg} for Circle-Circle-Circle trajectories}
The design of a Circle-Circle-Circle trajectory is accomplished by a suitable variation of $\{r_1,r_3,k\}$. All the reachable lengths are achieved by varying $k$ in $[-\pi/2,3\pi/2)$ and $r_1$ and $r_3$ in $\{-r_{\min},r_{\min}\}$. In this section, we explore the effect of a continuous variation of $r_1$ and $r_3$ in $\mathbb{R}\setminus(-r_{\min},r_{\min})$. Given a desired length $l_o$, we show that this results in the existence of multiple trajectories between the  oriented points $A$ and $B$. We begin the analysis with the following definition of a trajectory $\Tilde{\Lambda}$. 
\begin{definition}
    Given two oriented points $A$ and $B$, consider that the radii $r_1$ and $r_3$ take some fixed values in $\mathbb{R}$ such that eqn. \eqref{eq: existence of Hyperbola} holds. Let
    \begin{equation*}
        \Tilde{l}=\underset{k\in[-\pi/2,3\pi/2)}\min l(k)
    \end{equation*}
     be the length of the shortest trajectory in this framework. We denote this trajectory by $\Tilde{\Lambda}$.
\end{definition}
The trajectory $\Tilde{\Lambda}$ can be found numerically by iterating over $k\in[-\pi/2,3\pi/2)$. Note that it can be a $CSC$ trajectory which is a $\mathbf{C}^{r_1}\mathbf{C}^{r_2}\mathbf{C}^{r_3}$ trajectory with $|r_2|$ tending to ininity. We should highlight that for appropriate values of $(A,B,r_1,r_3)$ with $r_1,r_3\in\{\pm r_{\min}\}$, $\Tilde{\Lambda}={\Lambda_m}$. We define the following sets based upon the trajectory $\Tilde{\Lambda}$. 
\begin{align*}
        \mathcal{P}_0=\{(A,B,r_1,r_3)|&\Tilde{\Lambda}\in C_{\eta}^{r_1}C_{\mu}^{r_2}C_{\zeta}^{r_3}\text{ with } \mu\geq\pi|r_2|\}\\
        \mathcal{P}_1=\{(A,B,r_1,r_3)|&\Tilde{\Lambda}\in C_{\eta}^{r_1}C_{\mu}^{r_2}C_{\zeta}^{r_3}\text{ with } \eta\geq\pi|r_1|\}\\
        \mathcal{P}_2=\{(A,B,r_1,r_3)|&\Tilde{\Lambda}\in C_{\eta}^{r_1}C_{\mu}^{r_2}C_{\zeta}^{r_3}\text{ with } \zeta\geq\pi|r_3|\}\\
        \mathcal{P}_3=\{(A,B,r_1,r_3) | &\Tilde{\Lambda}\in C_{\eta}^{r_1}C_{\mu}^{r_2}C_{\zeta}^{r_3}\text{ with }\\ &d(\mathbf{c}^r_A,\mathbf{c}^r_B)\geq|r_1|+|r_3|+2r_{\min}\}\\
        \mathcal{P}_4=\{(A,B,r_1,r_3) |&\Tilde{\Lambda}\in C_{\eta}^{r_1}C_{\mu}^{r_2}C_{\zeta}^{r_3}\text{ with }\\ &d(\mathbf{c}^l_A,\mathbf{c}^l_B)\geq|r_1|+|r_3|+2r_{\min}\}
    \end{align*}
where $\eta$, $\mu$ and $\zeta$ are arc lengths of the three circles, respectively. The points $\mathbf{c}^r_A$ and $\mathbf{c}^l_A$ are the centres of the circles of radius $|r_1|$ at $A$ corresponding to right and left turns, respectively. The points $\mathbf{c}^r_B$ and $\mathbf{c}^l_B$ are the centres of the circles of radius $|r_3|$ at $B$ corresponding to right and left turns, respectively. We define the set $\mathcal{P}:=\mathcal{P}_1\cup\mathcal{P}_2\cup\mathcal{P}_3\cup\mathcal{P}_4$. The set $\mathcal{P}$ is a generalisation of the set $\mathcal{O}$ for different values of $r_1$ and $r_3$. The following theorem states the set of reachable lengths for these sets.
\begin{theorem}
    Given two oriented points $A$ and $B$, consider that the radii $r_1$ and $r_3$ take some fixed values in $\mathbb{R}$ such that eqn. \eqref{eq: existence of Hyperbola} holds. If $(A,B,r_1,r_3)\in\mathcal{P}_0\cup\mathcal{P}$, then there exists a $\mathbf{C}^{r_1}\mathbf{C}^{r_2}\mathbf{C}^{r_3}$ trajectory for all desired lengths $l_o\in[\Tilde{l},\infty)$.
    \label{thm: variable r1 and r3}
\end{theorem}
\begin{proof}
    The proof for the sets $\mathcal{P}_0$ and $\mathcal{P}$ is similar to that of Theorems \ref{lem: CCC ELONGATION} and \ref{lem: set O ELONGATION}, respectively.
\end{proof}

This concludes our discussion on the set of reachable lengths for curvature-bounded trajectories of desired lengths. Theorem \ref{thm: variable r1 and r3} provides set of reachable lengths for any given values of $r_1$ and $r_3$ in $\mathbb{R}$ through the construction of sets $\mathcal{P}_0$ and $\mathcal{P}$. 

\begin{remark}
    It should be noted that the reachability set in Theorem \ref{thm: variable r1 and r3} is always a subset of the maximum set of reachable lengths. Thus, it illustrates the existence of multiple trajectories for the same desired length. Note that this is in contrast to the results presented in Theorems \ref{lem: CCC ELONGATION} - \ref{lem: set OC elongation} which prove the existence of at least one $\mathbf{C}^{r_1}\mathbf{C}^{r_2}\mathbf{C}^{r_3}$ trajectory of a reachable desired length.  Further, all of the proposed elongation strategies result in curvature-bounded trajectories with at most two curvature discontinuities.
\end{remark}

\section{Numerical simulations}
\label{sec: simulation}
In this section, we illustrate the results presented in this paper through the following numerical simulations.

\textbf{Example 1:} 
Consider $A=(-3m,1m,0.785rad)$ and $B=(0m,0m,0rad)$. The Dubins Shortest Path between $A$ and $B$ for $r_{\min}=1m$ is an $RSL$ path of length $l_m=3.484m$. The given $(A,B)\in\mathcal{O}^c$ with $l_1=4.144m$ and $l_2=6.856m$ from eqn. \eqref{eq: l1 l2 definition}. Consequently, the set of reachable lengths can be defined. We seek to construct trajectories of lengths $l_o\in\{3.60m,4.05m,7.00m,11.15m,12.45m,14.90m\}$. Table \ref{tab: simulation 1} shows the computed values of $r_2$. Fig. \ref{fig: simulation 1} shows the various feasible trajectories. The Dubins Shortest Path has been highlighted in red.
\renewcommand{\arraystretch}{1.2}
\begin{table}[h]
\small
\caption{ Computed values of $r_2$ for trajectories between $A$ and $B$ of various lengths}
    \centering
    \begin{tabular}{ |c|c|c|c|c|c|}
    \hline
         $l_o(m)$&$r_1(m)$&$r_2(m)$&$k$&$r_3(m)$&Label\\
         \hline
         3.484 & -1.00 & $\infty$ & $\pi$/2& 1.00&{\color{red}\textbf{\textemdash}}\\
         \hline
         3.60 & -1.00 & -1.37 & 2.634& 1.00 & {\color{cyan}\textbf{\textemdash}}\\
         \hline
         4.05& 1.00& -1.031 & -0.379& 1.00&{\color{brown}\textbf{\textemdash}}\\
         \hline
         7.00 &1.00 & -1.015& 0.360& 1.00&{\color{violet}\textbf{\textemdash}}\\
         \hline
         11.15& 1.00& -1.57& 0.748& 1.00&{\color{magenta}\textbf{\textemdash}}\\
         \hline
         12.45& -1.00& 1.49& -0.634& 1.00&{\color{blue}\textbf{\textemdash}}\\
         \hline
         14.90& -1.00& 1.87& -0.876& 1.00&{\color{green}\textbf{\textemdash}}\\
         \hline   
    \end{tabular}
    \label{tab: simulation 1}
\end{table}
\begin{figure}[h]
    \centering
    \includegraphics[scale=0.46]{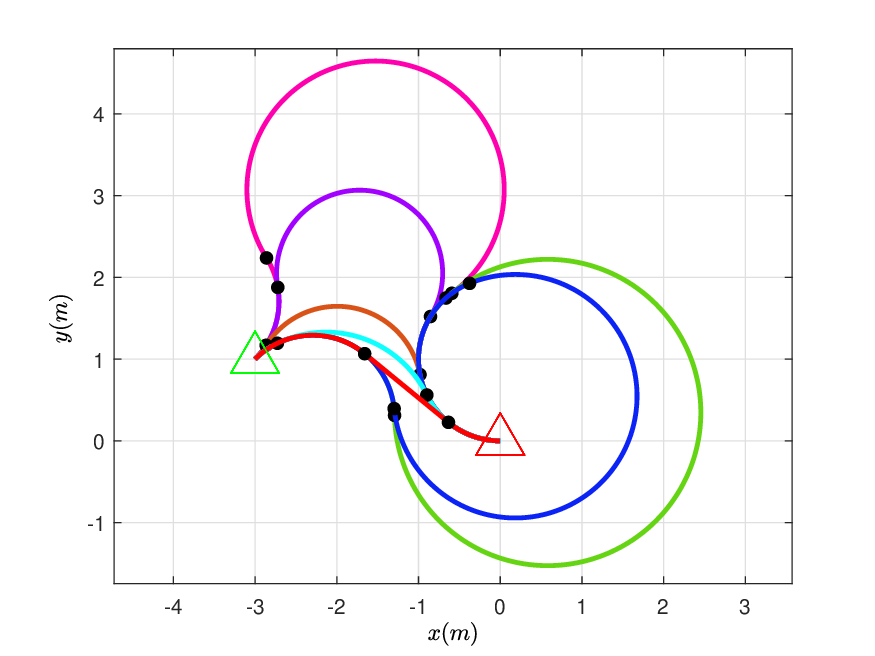}
    \caption{Trajectories between $A=(-3m,1m,0.785rad)$ and $B=(0m,0m,0rad)$ of various lengths with changeover points labeled ($\bullet$) }
    \label{fig: simulation 1}
\end{figure}

\textbf{Example 2:}
Consider $A=(-30m,10m,0.714rad)$ and $B=(0m,0m,0rad)$. The Dubins Shortest Path for $r_{\min}=1m$ is an $RSL$ path of length $l_m=31.809m$. The given $(A,B)\in\mathcal{O}_3$. Thus, a $\mathbf{C}^{r_1}\mathbf{C}^{r_2}\mathbf{C}^{r_3}$ trajectory exists for all $l_o\in[l_m,\infty)$. We seek to construct trajectories of length $l=44.5m>l_m$. There exists infinitely many such trajectories between $A$ and $B$. We arbitrarily choose $r_1$ and $r_3$ values and check if $(A,B,r_1,r_3)\in\mathcal{P}_0\cup\mathcal{P}$. If $l_o\geq\Tilde{l}$, then we proceed to compute the values of $r_2$. Table \ref{tab:r1 r2 r3 simulation} shows the computed values of $\Tilde{l}$ and $r_2$. Fig. \ref{fig:trajectpry ccc} shows the simulated trajectories in the $\mathbb{R}^2$ plane. Note that for the trajectory in green, $r_1$ and $r_3$ values are appropriately chosen such that we don't require three circular arcs resulting in a $LL$ trajectory.  

\begin{table}[h]
\small
    \centering
    \caption{ Computed values of $r_2$ and $k$ for trajectories between $A$ and $B$ of length $l_o=44.5m$}
    \label{tab:r1 r2 r3 simulation}
    \begin{tabular}{|c|c|c|c|c|c|}
    \hline
         $r_1(m)$&$r_3(m)$&$\Tilde{l}(m)$&$r_2(m)$&$k$&Label \\
         \hline
         -2.500 & 1.500& 32.099& 20.683&0.805&{\color{blue}\textbf{\textemdash}}\\
         \hline
         -5.500&-3.580&33.467&9.601&0.167& {\color{cyan}\textbf{\textemdash}}\\
         \hline
         -1.000&-1.010&32.389&14.798&3.328& {\color{magenta}\textbf{\textemdash}} \\
         \hline
         13.790&10.010&35.998&-9.145&-0.242&{\color{red}\textbf{\textemdash}} \\
         \hline
         1.940&12.010&35.673&-27.42&2.029&{\color{yellow}\textbf{\textemdash}} \\
         \hline
         2.040&59.314&N/A&N/A&N/A&{\color{green}\textbf{\textemdash}} \\
         \hline
    \end{tabular}
\end{table}
\begin{figure}[ht]
    \centering
    \includegraphics[scale=0.48]{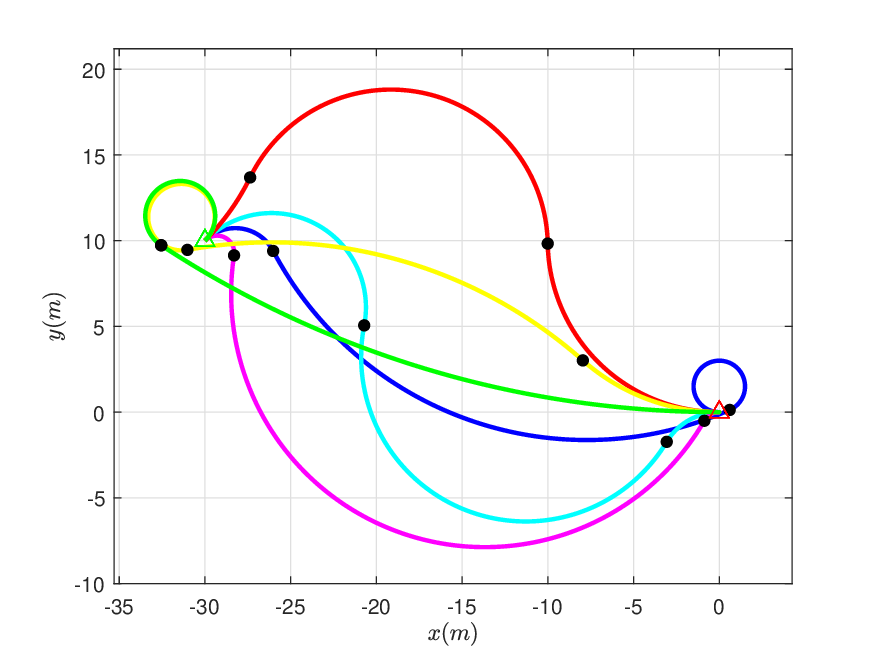}
    \caption{Trajectories between $A$ and $B$ of length $l_o=44.5m$ with changeover points labeled ($\bullet$)  }
    \label{fig:trajectpry ccc}
\end{figure}

\section{Conclusion}
\label{sec: conclusion}
In this paper, our objective is to construct curvature-bounded trajectories of any desired length between any two given oriented points. To do so, we propose to design the trajectory utilising three circular arcs of varying radii for the same, referred to as a Circle-Circle-Circle ($\mathbf{C}^{r_1}\mathbf{C}^{r_2}\mathbf{C}^{r_3}$) trajectory. The feasible trajectory is constructed by first fixing the terminal circles. Then, we show that the locus of the centre of $\mathcal{C}_2$ is a hyperbola $\mathcal{H}$ parameterised by the argument $k$. The overall trajectory is then defined using the three parameters: the radii of the terminal circles $\{r_1,r_3\}$ and an argument $k$. In the absence of any constraints on the length of the trajectory, we derive the necessary conditions for the existence of $\mathbf{C}^{r_1}\mathbf{C}^{r_2}\mathbf{C}^{r_3}$ trajectories. Now, such trajectories can be of eight types: $\{LLL,LLR,LRR,LRL,RRL,RLL,RLR,RRR\}$. We also present a complete classification of the trajectories into these forms based upon the values of $\{r_1,r_3,k\}$.

In the presence of curvature boundedness, we propose to elongate the circular arcs $C^{r_i}$, $\in\{1,2,3\}$ to achieve trajectories of desired lengths. In this regard, we show that the argument $k$ is critical as its variation, divided over the two branches of the hyperbola, results in an infinite (not necessarily continuous) elongation of the trajectory. However, this variation leads to jump discontinuities in the length of the trajectory for some configurations of $(A,B)$. To resolve this issue, we propose elongation strategies which guarantee the existence of curvature-bounded trajectories of any desired length for any configuration $(A,B)$. Further, we show that the set of reachable lengths is exactly equal to that proposed in \cite{chen_elongation} guaranteeing maximum coverage of the reachability set. In addition to this, the proposed elongation strategies also lead to the existence of multiple trajectories of desired lengths simply through the variation of $r_1$ and $r_3$. 

The paper concludes with numerical solutions illustrating and validating various results discussed in the paper. Future works in this direction may extend this approach to trajectory planning in the presence of obstacles. Further, additional constraints like minimum control effort maybe imposed as the trajectory is under-constrained for the problem addressed in this paper.

\bibliographystyle{ieeetr}
\bibliography{autosam.bib}
\end{document}